\documentclass[10pt,twocolumn, final]{IEEEtran}
\usepackage[dvips,final]{graphicx}
\usepackage{amsmath}
\usepackage{amsthm}
\usepackage{lscape}
\usepackage{latexsym}
\usepackage{amssymb}
\usepackage{bm}
\usepackage{bbm}
\usepackage{algorithmic}
\usepackage{algorithm}
\usepackage{color}

\usepackage{enumitem}
\usepackage[caption=false,font=footnotesize]{subfig}

\newtheorem{theorem}{Theorem}
\newtheorem{proposition}{Proposition}
\newtheorem{lemma}{Lemma}
\theoremstyle{definition}
\newtheorem{definition}{Definition}
\newtheorem{corollary}{Corollary}
\newcommand{\argmax}{\operatornamewithlimits{arg\,max}}
\newcommand{\argmin}{\operatornamewithlimits{arg\,min}}

\title{Cache-Aided Interactive Multiview Video Streaming in Small Cell Wireless Networks}

\author{Eirina~Bourtsoulatze, {\em Member, IEEE}, and Deniz~G\"{u}nd\"{uz}, {\em Senior Member, IEEE} %
\thanks{E.~Bourtsoulatze is with the Department of Electronic and Electrical Engineering at University College London, London WC1E 7JE, United Kingdom (e-mail: e.bourtsoulatze@ucl.ac.uk). 

D.~G\"{u}nd\"{uz} is with the Department of Electrical and Electronic Engineering at Imperial College London, London SW7 2AZ, United Kingdom  (e-mail: d.gunduz@imperial.ac.uk). } %
\thanks{This work has been funded by the European Union's Horizon 2020 research and innovation programme under the Marie Sklodowska-Curie grant agreement No. 750254 and by the European Research Council (ERC) under the European Union's Horizon 2020 research and innovation programme under grant agreement No. 677854. }
}

\begin{document}
\maketitle

\begin{abstract}
The emergence of novel interactive multimedia applications with high rate and low latency requirements has led to a drastic increase in the video data traffic over wireless cellular networks. Endowing the small base stations of a macro-cell with caches that can store some of the content is a promising technology to cope with the increasing pressure on the backhaul connections, and to reduce the delay for demanding video applications. In this work, delivery of an interactive multiview video to a set of wireless users is studied in an heterogeneous cellular network. Differently from existing works that focus on the optimization of the delivery delay and ignore the video characteristics, the caching and scheduling policies are  
jointly optimized, taking into account the quality of the delivered video and the video delivery time constraints. We formulate our joint caching and scheduling problem as the average expected video distortion minimization, and show that this problem is NP-hard. We then provide an equivalent formulation based on submodular set function maximization and propose a greedy solution with $\frac{1}{2}(1-\mbox{e}^{-1})$ approximation guarantee. The evaluation of the proposed joint caching and scheduling policy shows that it significantly outperforms benchmark algorithms based on popularity caching and independent scheduling. Another important contribution of this paper is a new constant approximation ratio for the greedy submodular set function maximization subject to a $d$-dimensional knapsack constraint.
\end{abstract}



\section{Introduction}
\label{sec:intro}

Recent analyses of visual networking applications indicate a steady growth in mobile video traffic, which is expected to reach more than three quarters of the total mobile traffic by 2021 \cite{CiscoVNI}. Large share of this traffic is attributed to video content generated by emerging multimedia applications, such as virtual reality (VR), augmented reality (AR) and interactive multiview video streaming (IMVS). These novel multimedia technologies offer users a completely new experience through the possibility to interact with the application in real time. However, high quality experience and interactivity comes with low latency and high bandwidth requirements that must be met by the mobile data operators. To deal with the ever increasing amount of mobile video data, the use of small cell base stations (SBSs) equipped with caches to store some of the high data rate content has been proposed in \cite{GolrezaeiINFOCOM2012}. SBS caches can be exploited by off-loading content to the caches during off-peak hours, and serving users locally through short-range low-latency communication during peak-hours. In that way, the use of costly backhaul links that connect the SBSs to the core network during the peak-hours can be alleviated and the load on the macro cell base station (MBS) can be reduced \cite{PoularakisTMC2017, PoularakisTCOMM2014}.  

Our work is motivated by the new challenges arising with the emergence of immersive and interactive multimedia technologies. In particular, we study an IMVS application provided to users over a cellular network. IMVS enables users to freely explore the 3D scene of interest from different viewpoints in real time \cite{DeAbreuJVCIR2015}. The challenge of offering such interactivity to users is the need to deliver multiple video streams corresponding to different views, as views selected by users are not known a priori. Thus, compared to single view conventional video streaming, an IMVS application typically requires much higher bandwidth to enable low-latency view switching at high quality. 

As in state-of-the-art wireless edge caching systems proposed for video-on-demand (VoD) applications \cite{GolrezaeiINFOCOM2012, PoularakisTMC2017,PoularakisTCOMM2014}, the placement of multiview video content in the SBS caches can bring the video content closer to wireless users, and reduce the load on the MBS and the backhaul links. As a result, a larger set of views can be delivered to users with a better quality of experience (QoE). However, the key objective in the context of caching for real-time video streaming is different from the one considered for VoD applications. In the latter, the users request a single file according to some popularity distribution and the aim is to place the video content in the caches in a way to minimize the average download delay. This objective is not suitable for real-time video streaming applications,  as it ignores the video delivery time constraints and the quality of the delivered content. 

In an IMVS system, users do not request a single file, but a set of views, which ideally would include all views captured by the cameras. However, when the available bandwidth is limited, only a subset of available views can be delivered to the users. Hence, when optimizing the caching policy one must take into consideration the quality of the delivered content, and perform the scheduling of optimal sets of views jointly with cache placement. Differently from existing caching solutions, we jointly optimize the caching and scheduling policies to ensure the delivery of the optimal sets of views within the time constraints imposed by the real-time video application. Though joint caching and routing has been previously considered in the literature \cite{PoularakisTCOMM2014, WangIEEEACCESS2017, KhreishahJSAC2016, OzfaturaCOMMLET2017}, the time constraints and the quality of the delivered video have not been taken into account in the proposed solutions.  

The optimal selection of the delivered views that maximizes the video quality has been studied in \cite{DeAbreuJVCIR2015,ToniMMSP2013}. In these works the delivered sets of views are optimized assuming the bandwidth resources of the users are fixed and known a priori. In a cellular network deploying multiple SBSs, each user may fall within the coverage range of several SBSs. Thus, there is no prior knowledge of the users' bandwidth capabilities. The users' bandwidth resources may vary depending on the density of the SBS placement and the number of users served by these SBSs, and, therefore, must be allocated jointly with the caching policy. 

\begin{figure}[t]
	\begin{center}
	\vspace{-1cm}
		\includegraphics[width = 0.3 \textwidth]{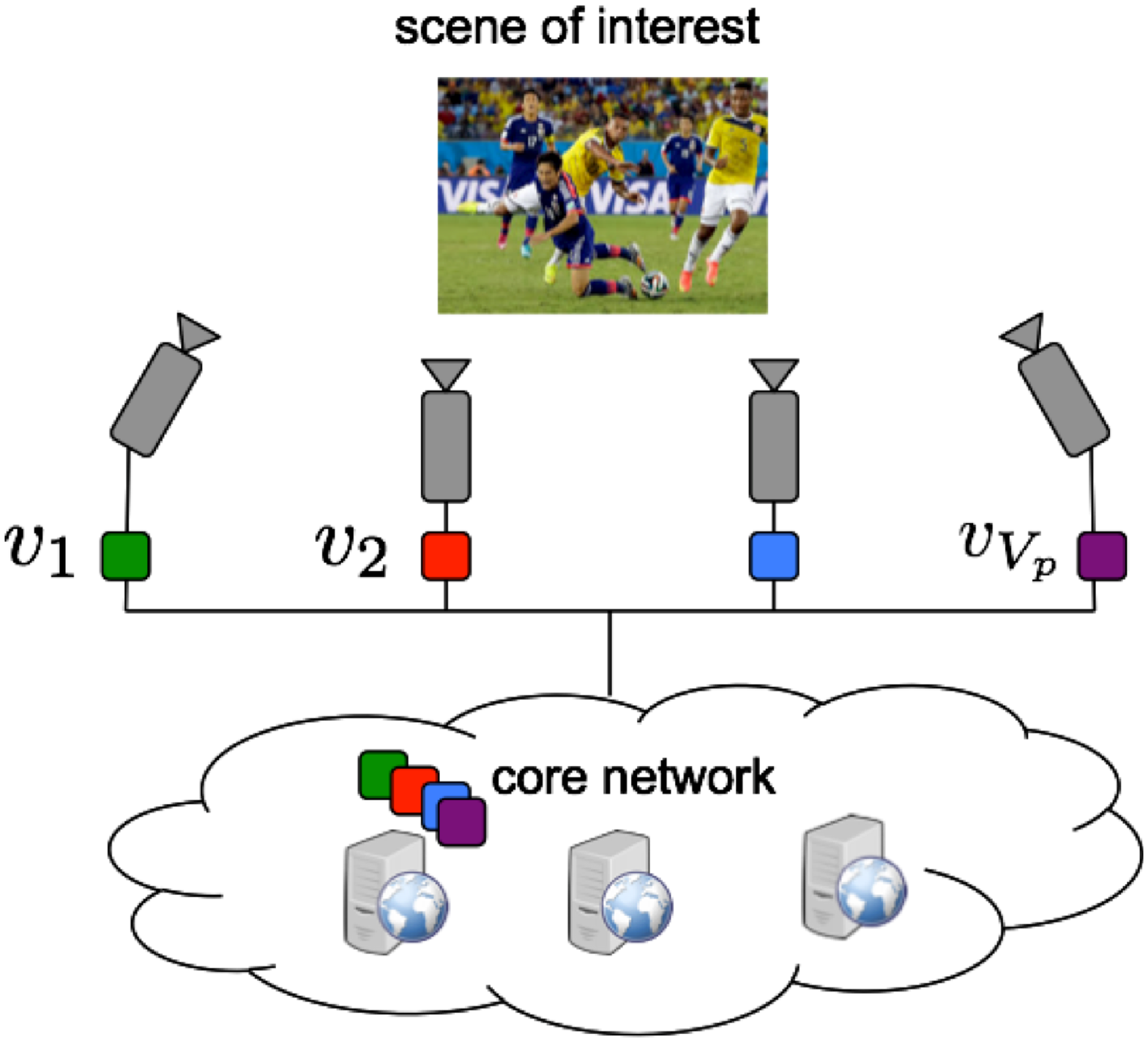}
	\vspace{-1cm}
	\end{center}
		\caption{Illustration of an IMVS with $V_p$ cameras capturing the scene of interest from multiple viewpoints. Captured content is transmitted to the core network.}
	\label{fig:imvs}
\end{figure}

In this paper, we propose a novel framework for jointly optimizing the caching and scheduling policies for the delivery of interactive multiview video in a wireless cellular network consisting of an MBS and multiple SBSs equipped with caches. The goal of our joint policy is to optimally allocate the cache and bandwidth resources of the wireless network in order to minimize the average expected distortion of the users who freely navigate through the available set of views during the streaming session. Differently from previous works, our framework takes into account the time constraints of the streamed video and the quality of the video content delivered to the users. Initially, we formulate our joint caching and scheduling problem for IMVS as the average expected distortion minimization and show that this problem is NP-hard. We then show that this problem can be equivalently expressed as the maximization of the reduction in the average expected distortion, and prove that the equivalent problem involves maximizing a submodular set function subject to a $d$-dimensional knapsack constraint. In order to efficiently solve our optimization problem, we adopt a greedy algorithm and prove that it admits a constant approximation ratio of $\frac{1}{2}(1-\mbox{e}^{-1})$. To the best of our knowledge, this is a novel result in the literature which extends the existing results on greedy submodular function optimization with a single knapsack constraint to the case with a $d$-dimensional knapsack constraint. Finally, we show through numerical evaluation that our proposed algorithm for joint caching and scheduling significantly outperforms benchmark algorithms based on popularity caching and independent rate allocation.


\section{System Model}
\label{sec:system}

We consider an IMVS application as illustrated in Fig.~\ref{fig:imvs}. An array of equally spaced cameras capture a 3D scene of interest from multiple viewpoints.  Let $\mathcal{V}_p \triangleq \{v_1, v_2, \dots, v_{V_p}\}$ denote the set of views acquired by the cameras, called the \textit{anchor views}, where $V_p = |\mathcal{V}_p| > 2$ is the total number of anchor views.  Each camera then encodes its acquired anchor view independently from the other cameras, and transmits it to the core network. This is a common assumption for distributed video acquisition scenarios, where the multiple cameras typically do not communicate with each other; and, thus, cannot jointly compress the captured video streams. In the core network, the compressed video streams are stored at the content provider's servers. 

From the core network, the video content is further delivered to a set of wireless users through a cellular network as shown in Fig.~\ref{fig:cellularnetwork}.  We study a streaming scenario where the IMVS session is initiated some time after the video content is recorded, and not immediately after the content is acquired as in live streaming. For example, this may correspond to the broadcasting of a sports event or a concert, which took place in one part of the world and could not be live-streamed to regions in different time zones. In this case, the content provider acquires the recorded content and makes it available according to a predefined viewing schedule, also providing an interactive multiview experience to users.

\begin{figure}[t]
	\begin{center}
	\vspace{-2.5cm}
			\includegraphics[width = 0.4 \textwidth]{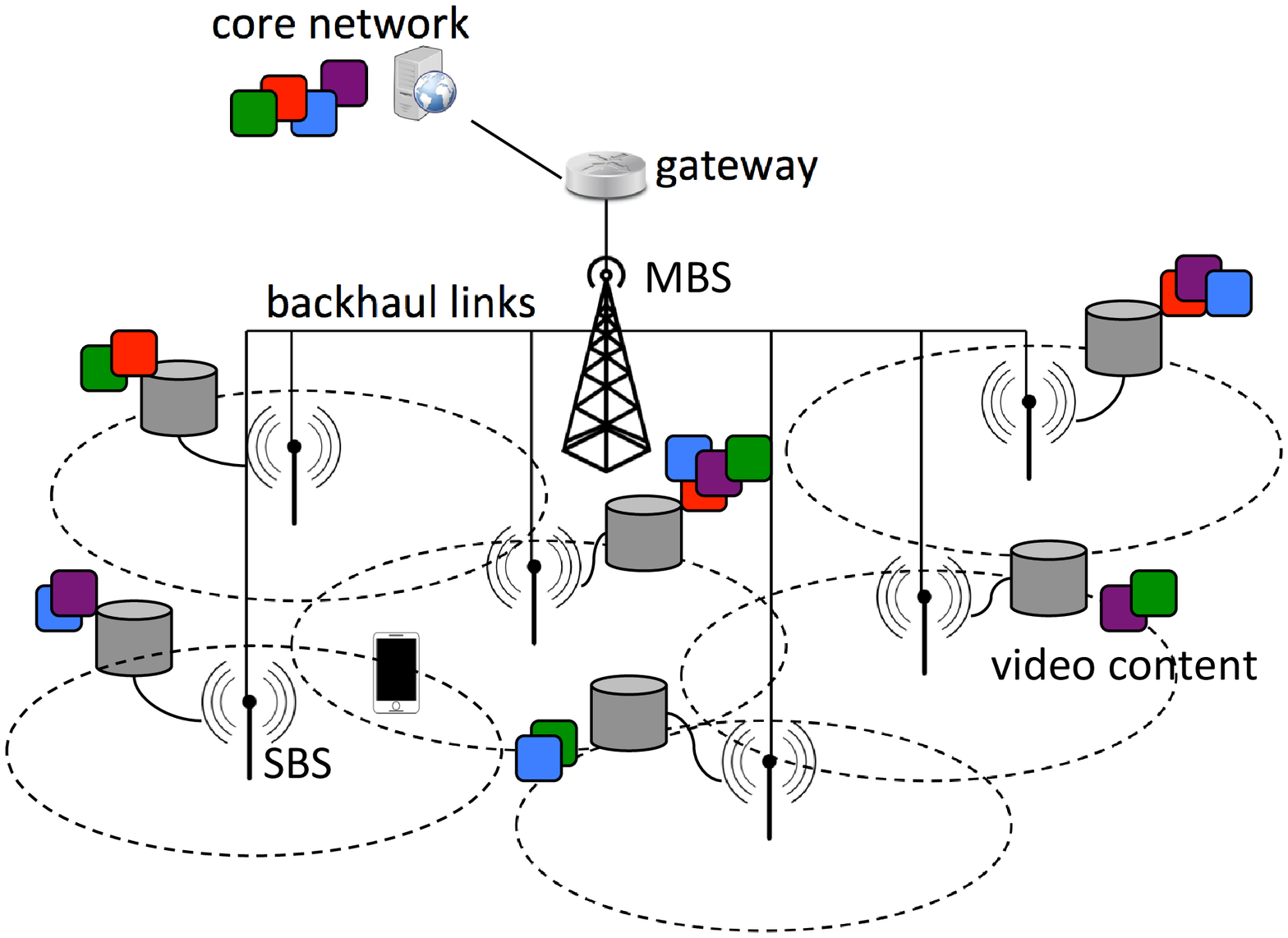}
	\vspace{-2cm}
	\end{center}
		\caption{The multiview video is cached in the small-cell base stations and delivered to wireless users simultaneously during the streaming session.}
	\label{fig:cellularnetwork}
\end{figure}

We focus on the operation of a single macro-cell that serves a set  $\mathcal{U} \triangleq \{1,2,\dots, U\}$ of $U$ wireless users. The macro-cell is covered by a total of $N+1$ base stations (BSs); an MBS indexed by $n=0$, and a set $\mathcal{N} \triangleq \{1,2,\dots, N\}$ of $N$ SBSs distributed across the macro-cell. The MBS covers the entire macro-cell region and can communicate with all the users within the cell. The SBSs have limited coverage and can only communicate with the users located within their proximity. Let $\mathcal{U}_n \subseteq \mathcal{U}$ denote the set of users that are covered by the $n$-th BS. Given a sufficiently dense placement of SBSs within the macro-cell area, every user will typically fall within the communication  range of multiple SBSs. We denote the set of BSs that can communicate with user $u$ as $\mathcal{N}_u \subseteq \mathcal{N}\cup\{0\}$. Note that, we have $\mathcal{U}_0 = \mathcal{U}$, and $0 \in \mathcal{N}_u,~ \forall u\in \mathcal{U}$. We assume that the MBS and the SBSs are assigned disjoint sets of subchannels, while neighbouring SBSs operate in orthogonal frequency bands \cite{GolrezaeiINFOCOM2012, PoularakisTCOMM2014}.  This permits us to ignore any interference among the BSs. Furthermore, we consider that SBS $n\in \mathcal{N}$ is equipped with a cache of size $C_n$ bytes and can store all or part of the multiview video, while the MBS is assumed to have unlimited cache capacity. Additionally, the total transmission capacity of each BS including the MBS is assumed to be limited, and equal to $R_n$ Mbps (Mbits per second), $\forall n \in \mathcal{N}\cup\{0\}$. This total transmission capacity $R_n$ must be allocated among the users within the coverage area of BS $n$, $\mathcal{U}_n$. To facilitate the caching and delivery of the multiview video content, and to use the network resources more efficiently, the video stream from each camera is partitioned into $T$ segments of $b^{t}$ bytes each,  $ t \in \mathcal{T} \triangleq \{1,2,\dots, T\}$.\footnote{We assume that the views are symmetrically coded. Hence, the size of a segment does not depend on the view index.} We denote the $t$-th segment of view $v$ as $B^{v,t}$. The index $t$ also stands for the time slot in which the segment $B^{v,t}$ can be scheduled for delivery.  

Since the video content is available at the server before its dissemination to the users, the caches available at SBSs can be exploited in order to reduce the load on the congested backhaul links during the streaming session. Part of the pre-recorded video content can be offloaded  to the caches of the SBSs during low-traffic hours and before the IMVS session starts. The cached content can then be directly delivered from the local caches to the users, avoiding thus the use of costly backhaul links and reducing the delivery delay which is critical in real-time video streaming.

In addition to the anchor views directly acquired by the cameras, additional virtual views can be synthesized from the texture and depth map information of the anchor views via the depth image based rendering (DIBR) technique \cite{SchmeingBOOK2011}. Specifically, we assume that at the decoder side,  $L$ additional virtual views can be synthesized between two adjacent views $v_i, \; v_{i+1} \in \mathcal{V}_p$. We denote the set of virtual views as $\mathcal{V}_s \triangleq \{v_{1+\delta}, v_{1+2\delta}, \allowbreak \dots, v_{1+L\delta}, \allowbreak \dots, v_{V_p - 1+\delta}, \allowbreak \dots,  v_{V_p - 1+L\delta}\}$, where $\delta = \frac{1}{L+1}$ is the distance between two adjacent views in the set $\mathcal{V}  \triangleq \mathcal{V}_p \cup \mathcal{V}_s$ of all available viewpoints. In order to virtually synthesize a view $v \in \mathcal{V}$, a left and a right reference view from the set of anchor views $\mathcal{V}_p$ are required. Let $v_l < v < v_r $  with $v_l$, $v_r$ $\in \mathcal{V}_p$ denoting the left and the right reference anchor views, respectively. The distortion at which the virtual view $v$ is synthesized depends generally on the quality of the reference anchor views $v_l, v_r$, and the spatial correlation between the virtual view and the anchor ones. For an array of aligned and equally spaced cameras, the spatial correlation is proportional to the relative distance between the synthesized view $v$ and the reference anchor views $v_l$ and $v_r$ \cite{DeAbreuJVCIR2015}. To measure the distortion of the synthesized view $v$, we adopt the distortion model proposed in \cite{RenARXIV2012}: 
\begin{equation}
d_v(v_l,v_r) =  \gamma \mbox{e}^{\alpha_v(v_r-v_l)}\Big(\mbox{e}^{\beta_v \mbox{min}\{ v-v_l, v_r-v\}}-1\Big),
\label{eq:dist_model}
\end{equation}
\noindent
where $d_v(v_l,v_r)$ is the distortion at which view $v$ can be reconstructed from reference anchor views $v_l$ and $v_r$. The parameters $\gamma$, $\alpha_v$, $\beta_v$ define the rate at which the distortion of the virtual view increases with the distance from the reference anchor views. Note that the optimization of the distortion function for virtual view synthesis is beyond the scope of this paper, and the model in \eqref{eq:dist_model} has been chosen due to its simplicity and accuracy. Our framework for optimizing the cache-aided IMVS in small cell wireless networks is general and can incorporate other distortion models as well.
  
At any given time during the IMVS session, a user can select any of the actual camera viewpoints in set $\mathcal{V}_p$, or choose to synthesize a virtual view from set $\mathcal{V}_s$ in real time. Hence, the user can freely navigate through the available set of views $\mathcal{V}$ and explore the 3D scene of interest from different viewpoints. To enable such interactivity at the best possible quality, the full set of anchor views $\mathcal{V}_p$ must be delivered to the user at any given time, so that the user can reconstruct any virtual view in set $\mathcal{V}_s$ from the best left and right reference anchor views. However, this is not always possible in a bandwidth-limited system due to the strict delay constraints imposed by the IMVS application. Depending on the available resources, typically only a subset of the anchor views can be delivered to the wireless users. 

The subset of anchor views delivered to the user determines the distortion at which viewpoints selected by the user and not included in the delivered set of anchor views can be reconstructed. Note that the users select the desired viewpoints in real time. This implies that the views requested by a user at any given time during an IMVS session are not known when the contents are placed into the caches of SBSs. Thus, the subset of views to be stored at the SBSs has to be optimized with respect to the expected video quality based on a probabilistic model of the popularity of video segments of each view. For each segment $B^{v,t}$, we define a popularity $p^{v,t} \in [0,1]$ which represents the probability that the $t$-th segment of view $v$ will be selected for viewing. The segment popularity depends on the video content corresponding to this segment and without loss of generality (w.l.g.) can be considered the same for all users. The content popularity can be learned by the content provider by analyzing the multimedia content \cite{TzelepisIVC2016}, or the history of viewing requests \cite{BlascoICC14}. To guarantee the reconstruction of any view within set $\mathcal{V}$ at a minimum quality, we consider that views $v_1$ and $v_{V_p}$, the leftmost and rightmost anchor views, are always delivered to all the users by the MBS.\footnote{A potential alternative for the delivery of high-rate multiview video content is \textit{hybrid networking}, which combines terrestrial digital video broadcasting (DVB) with broadband cellular networks \cite{EkmekciogluTCSVT2017}. Therefore, one can consider that the two extreme anchor views are broadcasted through the terrestrial network to all the users.} The remaining transmission capacity $R_0$ of the MBS, as well as the cache capacity $C_n$ and transmission capacity $R_n$ of the SBSs are used to deliver additional anchor views to further improve the quality of the synthesized views.

In this work, we aim to find the optimal joint caching and scheduling policy for the multiview video segments $B^{v,t}$ that minimizes the average expected distortion at the wireless users which participate in the IMVS session. In the next section, we provide the formal problem formulation, and prove that the joint caching and scheduling problem for IMVS is NP-hard.


\section{Joint Caching and Scheduling problem}
\label{sec:problem}

\subsection{Problem formulation}

Let us define the binary variable $x_n^{v,t} \in \{0,1\}$,  where $x_n^{v,t} = 1$, if the $t$-th segment of view $v \in \mathcal{V}_p$ is placed in the cache of the $n$-th SBS during the caching phase, and $x_n^{v,t} = 0$ otherwise. Similarly, let $y_{n,u}^{v,t} \in \{0,1\}$ be a binary decision variable which indicates whether the segment $B^{v,t}$ is scheduled for delivery from BS $n$ to user $u$; $y_{n,u}^{v,t} =1$  if the segment $B^{v,t}$ is scheduled, and $y_{n,u}^{v,t} =0$ otherwise. The vector $(\bm x, \bm y)$ of decision variables, where 
\begin{equation}
\begin{split}
\bm x &\triangleq (x_{n}^{v,t} \in \{0,1\},\; \forall n\in \mathcal{N}, \; v \in \mathcal{V}_p, \; t\in \mathcal{T}), \; \mbox{and} \\
\bm y &\triangleq ( y_{n,u}^{v,t}\in \{0,1\}, \;  \forall n \in \mathcal{N}\cup\{0\}, \; u \in \mathcal{U}, \; v\in \mathcal{V}_p, \; t\in \mathcal{T} ),
\end{split}
\end{equation}
\noindent
defines a joint caching and scheduling policy. From the assumption that the leftmost and rightmost anchor views $v_1$ and $v_{V_p}$, respectively,  are always delivered to the users by the MBS, we have 
\begin{equation} 
\begin{split}
&x_n^{v,t} = 0,\; \forall n \in \mathcal{N}, \;  v \in \{v_1, v_{V_p}\}, \;  t \in \mathcal{T}, \\
&y_{n,u}^{v,t} = 0, \; \forall n \in \mathcal{N},\;   u \in \mathcal{U}, \;  v \in \{v_1, v_{V_p}\}, \;  t \in \mathcal{T}, \\
&y_{0,u}^{v,t} = 1, \; \forall u \in \mathcal{U}, \;  v \in \{v_1, v_{V_p}\}, \;  t \in \mathcal{T}.  \\
\end{split}
\label{eq:assumption}
\end{equation}

The distortion $D_{u}^{v,t}(\bm x, \bm y)$ at user $u$ for reconstructing segment $B^{v,t}$ for the given caching and scheduling policy $(\bm x, \bm y)$  can be expressed as follows:

\begin{equation}
D_{u}^{v,t}(\bm x, \bm y) =
\begin{cases}
\tilde{D}_{u}^{v,t}(\bm x, \bm y)(1- \mathbbm{1}_{\{\sum_{n\in\mathcal{N}_u }y_{n,u}^{v,t} > 0\}}), \quad \forall v \in \mathcal{V}_p \\
\tilde{D}_{u}^{v,t}(\bm x, \bm y) ,\quad \forall v \in \mathcal{V}_s
\end{cases}, 
\label{eq:dist_u_v_t}
\end{equation}
\noindent
 where  $\tilde{D}_{u}^{v,t}(\bm x, \bm y)$ is the minimum distortion at which segment  $B^{v,t}$ can be reconstructed at user $u$ given the joint caching and scheduling policy $(\bm x, \bm y)$. The indicator function $ \mathbbm{1}_{\{c\}}$ is ``1'' if the condition $c$ is true, and ``0'' otherwise. In Eq.~\eqref{eq:dist_u_v_t}, we distinguish the following two cases. When view $v$ belongs to the set of anchor views captured by the cameras, the distortion for reconstructing the segment $B^{v,t}$  at user $u$  is 0, if the segment is delivered to user $u$ by at least one of the BSs $n\in \mathcal{N}_u$ that cover user $u$. Otherwise, the distortion is equal to the minimum distortion $\tilde{D}_{u}^{v,t}(\bm x, \bm y)$ at which segment $B^{v,t}$ can be reconstructed at user $u$ given the joint caching and scheduling policy $(\bm x, \bm y)$. When view $v$ is a virtual view, the segment $B^{v,t}$ is not delivered to user $u$, and is synthesized using the corresponding segments of the closest left and right anchor views according to the joint caching and scheduling policy $(\bm x, \bm y)$.  Finally, the minimum distortion $\tilde{D}_{u}^{v,t}(\bm x, \bm y)$ at which segment $B^{v,t}$ can be reconstructed at user $u$ when it is not delivered,  is given by
 
 \begin{equation}
 \begin{split}
 \tilde{D}_{u}^{v,t}  (\bm x, \bm y) &= \sum_{v_l < v} \sum_{v_r > v} d_v(v_l,v_r) \cdot  \\
 &\mathbbm{1}_{\{\sum_{n \in \mathcal{N}_u} y _{n,u}^{v_l,t} > 0\}}   \prod_{v_l < v_l^\prime < v} \Big(1 -   \mathbbm{1}_{\{\sum_{n \in \mathcal{N}_u} y_{n,u}^{ v_l^\prime,t}>0\}}\Big) \cdot \\
		& \mathbbm{1}_{ \{\sum_{n \in \mathcal{N}_u} y _{n,u}^{v_r,t} > 0\}}   \prod_{v< v_r^\prime < v_r} \Big(1 - \mathbbm{1}_{\{\sum_{n \in \mathcal{N}_u} y_{n,u}^{ v_r^\prime,t}>0\}}\Big)
		\end{split}
		\label{eq:min_dist}
 \end{equation}
 
\noindent
The second term of the product in Eq.~\eqref{eq:min_dist} is equal to ``1'' if for the view index $v_l \in \mathcal{V}_p$ the segment $B^{v_l,t}$ is delivered to user $u$ by at least one BS, and for all other segments $B^{v_l^\prime,t}$ delivered to user $u$, the view $v^\prime_l$ is farther from $v$ than $v_l$. Similarly, the third term of the product in Eq.~\eqref{eq:min_dist} is equal to ``1'' if for the view index $v_r \in \mathcal{V}_p$ the segment $B^{v_r,t}$ is delivered to user $u$ by at least one BS, and for all other segments $B^{v_r^\prime,t}$ delivered to user $u$, the view $v^\prime_r$ is farther from $v$ than $v_r$. Note that the product of the second and the third terms of the product in Eq.~\eqref{eq:min_dist} is non-zero only for one unique pair of left and right views $v_l$ and $v_r$. Furthermore, due to the assumption that all segments for views $v_1, \, v_{V_p} $ are always delivered to all the users, one such pair always exists. Finally, from Eqs. ~\eqref{eq:dist_u_v_t} and \eqref{eq:min_dist} we can observe that the distortion $D_{u}^{v,t}(\bm x, \bm y)$ is a function of only the scheduling vector $\bm y$. However, the distortion depends implicitly on the caching policy $\bm x$ since the latter determines the feasible schedules at SBSs. Thus, any attempt at minimizing the distortion must jointly consider the caching and scheduling decisions. 

Our goal is to devise a joint caching and scheduling policy that minimizes the average expected distortion of the wireless users. The optimization problem can be formally written as 

\begin{equation}
(\bm x^*, \bm y^*) = \argmin_{(\bm x, \bm y)} \frac{1}{U} \frac{1}{T} \sum_{u \in \mathcal{U}}\sum_{t\in\mathcal{T}}\sum_{v\in \mathcal{V}} D_{u}^{v,t}(\bm x, \bm y)p^{v,t} 
\label{eq:opt_prob}
\end{equation}
\begin{equation} 
\mbox{s.t.} \quad  \sum_{t\in\mathcal{T}} \sum_{v\in \mathcal{V}_p \backslash \{v_1, v_{V_p}\}} x_n^{v,t}b^t \leq C_n, \; \forall n \in \mathcal{N},
 \label{eq:caching}
 \end{equation}
 \begin{equation} 
\sum_{u\in\mathcal{U}} \sum_{v\in \mathcal{V}_p \backslash \{v_1, v_{V_p}\}} y_{n,u}^{v,t} r \leq R_n, \; \forall n\in\mathcal{N}\cup\{0\}, \; t \in \mathcal{T} ,
\label{eq:rate}
\end{equation} 
\begin{equation} 
y_{n,u}^{v,t} \leq x_n^{v,t}, \quad \forall u \in \mathcal{U}, \;  n \in \mathcal{N}, \;  v \in \mathcal{V}_p, \;  t \in \mathcal{T}.
\label{eq:scheduling}
\end{equation}
\begin{equation} 
x_n^{v,t}, y_{n,u}^{v,t} \in \{0,1\} \;\; \mbox{and constraints in \eqref{eq:assumption}}
\label{eq:additional_constraints}
\end{equation}

Constraint \eqref{eq:caching} is the cache capacity constraint and guarantees that the total amount of data stored in an SBS's cache does not exceed its capacity. Constraint \eqref{eq:rate} is the transmission capacity constraint, which states that the total rate delivered by an SBS in time slot $t$ must not exceed its capacity. The constant $r$ denotes the video rate, and since we consider symmetric coding of the views, it is the same for all captured views. Finally, the inequality in \eqref{eq:scheduling} couples the caching and scheduling decisions, and ensures that only the segments that are stored in SBS caches can be scheduled for transmission. 
 
The optimization problem defined in Eqs.~\eqref{eq:opt_prob}-\eqref{eq:additional_constraints} is an integer program which is difficult to solve directly due to the non-convex nature of the objective function and the integer constraints. In the next subsection, we show that this problem belongs to the class of NP-hard problems.


\subsection{Complexity}
\label{sec:complexity}

We now show that the optimization problem in \eqref{eq:opt_prob}-\eqref{eq:scheduling} is NP-hard. To prove that, it is sufficient to show that the corresponding decision version of the problem, which we call the \textit{joint caching and scheduling (JCS)} decision problem, is NP-hard. The proof relies on the method of restriction, which consists of placing additional restrictions on the instances of a given problem ${\bm P} \in$ NP so that the restricted problem is equivalent to some known NP-complete problem ${\bm P}^\prime$ \cite{GareyBOOK1979}.  In the following, we first define the decision version of our optimization problem, and provide the definition of the \textit{set K-cover (SKC)} decision problem \cite{SlijepcevicICC2001}, which will be used in the proof. 

\begin{definition}
\textit{JCS decision problem}: Given the set of BSs $\mathcal{N} \cup \{0\}$, the set of users $\mathcal{U}$, the sets of views $\mathcal{V}_p $ and $ \mathcal{V}_s$, the segment size values $\mathcal{B} = \allowbreak \{b^1, b^2, \dots, b^T\}$, the number of segments $T$, the segment popularity $\mathcal{P} = \{p^{v,t}\}$,  the SBSs' cache capacity values $\mathcal{C}  =\allowbreak \{C_1,  \allowbreak C_2, \dots, C_N\}$, the BSs' total transmission capacity values $\mathcal{R} = \{R_0, R_1, \dots, R_N\}$, the distortion function $D_{u}^{v,t}(\bm x, \bm y)$ and a positive real number $Z$, determine if there exists a  feasible joint caching and scheduling policy $({\bm x}, {\bm y})$ that satisfies the constraints in \eqref{eq:caching}-\eqref{eq:additional_constraints} and 
\begin{equation}
\frac{1}{U}\frac{1}{T}\sum_{u \in \mathcal{U}} \sum_{t \in \mathcal{T}}\sum_{v \in \mathcal{V}} D_{u}^{v,t}(\bm x, \bm y) p^{v,t}  \leq Z.
\end{equation}
\end{definition}
\noindent We denote the JCS decision problem instance as JCS($\mathcal{N} \cup\{0\}, \mathcal{U},  \mathcal{V}_p, \mathcal{V}_s, \mathcal{B}, T, \mathcal{P}, \mathcal{C}, \mathcal{R}, D_{u}^{v,t}(\bm x, \bm y),\allowbreak Z$). 

\begin{definition}
\textit{SKC decision problem \cite{SlijepcevicICC2001}}: Given a collection of subsets $\mathcal{S}$ of a set $\mathcal{A}$ and a positive integer $K \geq 2$, does $\mathcal{S}$ contain $K$ disjoint covers  for $\mathcal{A}$, {\em i.e.}, covers $\mathcal{S}_1, \mathcal{S}_2, \dots, \mathcal{S}_K$, where $\mathcal{S}_k \subset S$, such that every element of $\mathcal{A}$ belongs to at least one member of each of $\mathcal{S}_k$? 
\end{definition}
\noindent We denote the above SKC decision problem by SKC$(\mathcal{A}, \mathcal{S}, K)$. The SKC decision problem is known to be NP-complete \cite{SlijepcevicICC2001}. 
\begin{proposition}
The JCS decision problem is NP-hard.
\end{proposition}

\begin{proof}
Consider the SKC($\mathcal{A},\mathcal{S}, K$) decision problem and an instance of the JCS($\mathcal{N} \cup\{0\},  \mathcal{U},  \mathcal{V}_p,\mathcal{V}_s, \mathcal{B},T,  \allowbreak\mathcal{P}, \mathcal{C}, \mathcal{R}, D_{u}^{v,t}(\bm x, \bm y), Z$) decision problem with $|\mathcal{V}_p| -2 = K$, $\mathcal{U} = \mathcal{A}$, $\{\mathcal{U}_1,\dots, \mathcal{U}_N\} = \mathcal{S}$,  $T = 1$, $\mathcal{B} = \{b\}$, $p^{v,t} = \frac{1}{|\mathcal{V}|}$, $\mathcal{C} = \{b,b, \dots, b\}$, $R_0 = 0$, $R_n = |\mathcal{U}_n|$ and $Z =\allowbreak  \frac{|\mathcal{V}_s|}{|\mathcal{V}|}\sum_{v \in \mathcal{V}_s} D_{min}^{v} $, where $D_{min}^v$ is the minimum distortion at which a virtual view $v \in \mathcal{V}_s$ can be synthesized, and it is achieved when view $v$ is reconstructed from the two closest left and right reference anchor views. This instance corresponds to the scenario in which the whole stream consists of a single segment $(T=1)$ of size $b$.  Each SBS can cache only one single view ($C_n = b$), and the users receive data only from the SBSs ($R_0 = 0$). The $n$-th SBS can deliver data to all the users in the set $\mathcal{U}_n$ simultaneously ($R_n = |\mathcal{U}_n|$). Furthermore, every user requests one view from the set $\mathcal{V}$ of available views uniformly at random ($p^{v,t} = \frac{1}{|\mathcal{V}|}$). The average expected distortion for this instance of the JCS decision problem is lower bounded by 

\begin{equation}
\begin{split}
\frac{1}{U}\frac{1}{T}\sum_{u\in \mathcal{U}} \sum_{t \in \mathcal{T}}\sum_{v \in \mathcal{V}} D_{u}^{v,t}(\bm x, \bm y) p^{v,t}  \geq    \frac{|\mathcal{V}_s|}{|\mathcal{V}|}\sum_{v \in \mathcal{V}_s} D_{min}^{v}, \; \forall~ (\bm x, \bm y),
\end{split}
\end{equation}

\noindent which follows immediately if we observe that $D_{u}^{v,t}(\bm x, \bm y)  \geq 0$, $\forall v \in \mathcal{V}_p$, and $D_{u}^{v,t}(\bm x, \bm y) \geq D_{min}^v$, $\forall v \in \mathcal{V}_s$. The minimum value of the average expected distortion can only be attained  if every captured view $v \in \mathcal{V}_p \backslash \{v_1, v_{V_p}\}$ can be delivered to every user $u \in \mathcal{U}$ by at least one of the SBSs in set  $\mathcal{N}_u$ that cover user $u$. Therefore, deciding whether there exists a joint caching and scheduling policy ($\bm x, \bm y$), such that the average expected distortion is equal to $Z$, reduces to determining whether there exists a caching policy $\bm x$ such that every view in set $\mathcal{V}_p \backslash \{v_1, v_{V_p}\}$ is cached in at least one SBS in set $\mathcal{N}_u$ for every user $u$, since, due to the assumption that $R_n = |\mathcal{U}_n|$, the view can always be delivered. This, in turn, is equivalent to finding $|\mathcal{V}_p| -2 = K$ disjoint subsets of SBSs, such that every user is covered by at least one SBS in each subset. It now becomes apparent that the considered instance of the JCS decision problem is equivalent to the SKC decision problem, which is known to be NP-complete. It, therefore, follows that the JCS decision problem is NP-hard.
\end{proof}


\section{Expected distortion reduction maximization}
\label{sec:solution}

In order to deal with the computational complexity of the optimization problem in \eqref{eq:opt_prob}-\eqref{eq:additional_constraints}, we reformulate it as an equivalent problem which aims at maximizing the average expected distortion reduction. We express the equivalent optimization problem as a maximization of a set function defined over subsets of an appropriately selected ground set. We then show that the objective set function is a monotone non-decreasing submodular function. This permits us to devise efficient solutions based on the greedy approach.

\subsection{Equivalent problem formulation}

Let us  define the ground set $\mathcal{E}$ as  
\begin{equation} 
\begin{split}
\mathcal{E} \triangleq \{e_{n, \mathcal{A}_n}^{v,t}: \;  n \in \mathcal{N}\cup  \{0\},  \mathcal{A}_n \subseteq \mathcal{U}_n,   v \in \mathcal{V}_p\backslash \{v_1, v_{V_p}\}, \;   t \in \mathcal{T} \}
\end{split}
\label{eq:ground_set}
\end{equation}

\noindent The element $e_{n, \mathcal{A}_n}^{v,t}$ of the ground set $\mathcal{E}$ denotes the placement of the segment $B^{v,t}$ in the cache of BS $n\in \mathcal{N}\cup \{0\}$ and its scheduling for delivery to a subset $\mathcal{A}_n \subseteq \mathcal{U}_n \subseteq \mathcal{U}$ of the users covered by BS $n $. Any joint caching and scheduling policy $(\bm x, \bm y)$ can be represented by a subset $\mathcal{S}$ of the ground set $\mathcal{E}$. For example, placing the element $e_{n, \mathcal{A}_n}^{v,t} \in \mathcal{E}$ in $\mathcal{S}$ can be regarded as setting the decision variables $x_n^{v,t}$ and $y_{n,u}^{v,t}, \; \forall u \in \mathcal{A}_n$, to ``1''.  Recall that, by our initial assumption, all segments of views $v_1$ and $v_{V_p}$ are delivered by the MBS to all the users. Thus, $e_{0,{\mathcal{U}_0}}^{v_1,t}$ and $e_{0,{\mathcal{U}_0}}^{v_{V_p},t}$, $\forall t \in \mathcal{T}$, will always be included in a set $\mathcal{S}$ that represents a joint caching and scheduling policy $(\bm x, \bm y)$. 

Let us further define the sets $\mathcal{F}_u^{v,t} \triangleq \{e_{n, \mathcal{A}_n}^{v,t}:  n \in \mathcal{N}_u,\; \mathcal{A}_n \subseteq\mathcal{U}_n \; \mbox{s.t. } u\in \mathcal{A}_n\}$, $\forall u \in \mathcal{U}$, $\forall v \in \mathcal{V}_p$, $\forall t \in \mathcal{T}$. Set $\mathcal{F}_u^{v,t}\subseteq \mathcal{E}$ essentially represents all possible ways to deliver segment $B^{v,t}$ to user $u$. Given the ground set $\mathcal{E}$  and sets $\mathcal{F}_u^{v,t}$, we can re-write the distortion function in \eqref{eq:opt_prob} in the form of a set function $ D_u^{v,t}(\mathcal{S}): 2^\mathcal{E} \rightarrow \mathbb{R}$ as follows:
\begin{equation}
\begin{split}
 D_u^{v,t}(\mathcal{S}) = 
  \begin{cases}
\tilde{D}_u^{v,t}(\mathcal{S})\Big( 1- \mathbbm{1}_{\{\mathcal{S} \cap \mathcal{F}_u^{v,t} \neq \emptyset\}}  \Big) ,  \; \quad  v \in \mathcal{V}_p\\
\tilde{D}_u^{v,t}(\mathcal{S}),   \quad  ~ \; v \in \mathcal{V}_s
\end{cases},
\end{split}
\label{eq:dist_set_fuc}
\end{equation}

\noindent 
where 
\begin{equation}
\begin{split}
\tilde{D}_u^{v,t}(\mathcal{S}) =  &
 \sum_{v_l < v} \sum_{v_r > v}d _v(v_l,v_r) \cdot \\
& \mathbbm{1}_{\{\mathcal{S} \cap \mathcal{F}_u^{v_l,t} \neq \emptyset \}} 
 \prod_{v_l < v_{l^\prime} < v} (1-\mathbbm{1}_{\{\mathcal{S} \cap \mathcal{F}_u^{v_{l^\prime},t} \neq \emptyset \}}) \cdot\\
 & \mathbbm{1}_{\{\mathcal{S} \cap \mathcal{F}_u^{v_r,t} \neq \emptyset \}} 
 \prod_{v <  v_{r^\prime} < v_r} (1-\mathbbm{1}_{\{\mathcal{S} \cap \mathcal{F}_u^{v_{r^\prime},t} \neq \emptyset \}}).
\end{split}
\label{eq:dist_tilde}
\end{equation}

\noindent 
The distortion reduction at user $u$ for reconstructing the segment $B^{v,t}$ is defined as 
\begin{equation}
\Delta D^{v,t}_{u}(\mathcal{S}) = D_{max} - D_{u}^{v,t}(\mathcal{S}), \quad \forall \mathcal{S} \subseteq \mathcal{E},
\label{eq:dist_red_set_fun}
\end{equation}
where $D_{max}$ is the maximum distortion when the corresponding segment cannot be reconstructed. The distortion reduction function in Eq.~\eqref{eq:dist_red_set_fun} represents the reduction in the distortion experienced by user $u$ after reconstructing the segment $B^{v,t}$.

The constraints defined in Eqs.~\eqref{eq:caching},~\eqref{eq:rate},~\eqref{eq:scheduling} can also be expressed in terms of set functions defined over the ground set $\mathcal{E}$. Recall that the element $e_{n,\mathcal{A}_n}^{v,t}$ represents the joint placement of the segment $B^{v,t}$ in the cache of BS $n$ and its delivery to a subset of users $\mathcal{A}_n$. This implies that, when the element  $e_{n,\mathcal{A}_n}^{v,t}$  is included in the solution set $\mathcal{S}$, the segment $B^{v,t}$ is placed in the cache of BS $n$ consuming a total space of $b^{t}$ bytes and a rate of $|\mathcal{A}_n|r$ Mbps is allocated by the BS $n$ to transmit it to the users in $\mathcal{A}_n$.  Thus, with each element $e_{n,\mathcal{A}_n}^{v,t} \in \mathcal{E}$, we associate a caching cost of $b^{t}$ bytes and a rate cost of $|\mathcal{A}_n|r$ Mbps.  We define the cache cost and rate cost functions $c_n(\mathcal{S}): 2^\mathcal{E} \rightarrow \mathbb{R}$ and $r_{n}^{t}(\mathcal{S}): 2^\mathcal{E} \rightarrow \mathbb{R}$, respectively, as:

\begin{equation}
\begin{split}
c_{n}(\mathcal{S}) =  \sum_{e_{n^\prime,\mathcal{A}_{n^\prime}}^{v,t}\in\mathcal{S}} & c_n(e_{n^\prime,\mathcal{A}_{n^\prime}}^{v,t}),  \; \forall n\in\mathcal{N}, 
\end{split}
\label{eq:cache_cost_fn}
\end{equation}

\begin{equation}
\begin{split}
r_{n}^{t}(\mathcal{S}) =  \sum_{e_{n^\prime,\mathcal{A}_{n^\prime}}^{v,t^\prime}\in\mathcal{S}} & r_n^t(e_{n^\prime,\mathcal{A}_{n^\prime}}^{v,t^\prime}), \; \forall n\in\mathcal{N}\cup\{0\},
\end{split}
\label{eq:rate_cost_fn}
\end{equation}

\noindent
where  $c_n(e_{n^\prime,\mathcal{A}_{n^\prime}}^{v,t}) = b^{t}$ if $n^\prime = n$, and 0 otherwise, and $r_n(e_{n^\prime,\mathcal{A}_{n^\prime}}^{v,t}) = |\mathcal{A}_n|r$ if $n^\prime = n, t^\prime = t$, and 0 otherwise.

\noindent
Finally, we define the cost function $f_n^{v,t}(\mathcal{S}): 2^\mathcal{E} \rightarrow \mathbb{R}$ as:
\begin{equation}
f_n^{v,t} (\mathcal{S})= \sum_{e_{n^\prime,\mathcal{A}_{n^\prime}}^{v^\prime,t^\prime}\in\mathcal{S}} f_n^{v,t} (e_{n^\prime,\mathcal{A}_{n^\prime}}^{v^\prime,t^\prime}),
\label{eq:counter_fn}
\end{equation}

\noindent
where  $f_n^{v,t} (e_{n^\prime,\mathcal{A}_{n^\prime}}^{v^\prime,t^\prime}) = 1$ if $n^\prime = n, v^\prime = v, t^\prime = t$, and 0 otherwise. Essentially, function $f_n^{v,t}(\mathcal{S})$ counts the number of times segment $B^{v,t}$ is placed in the cache of BS $n$.

We can now reformulate the minimization problem in Eqs.~\eqref{eq:opt_prob}-\eqref{eq:additional_constraints} as a maximization of the average expected distortion reduction as follows: 
\begin{equation}
\mathcal{S}_{OPT} = \underset{\mathcal{S}}{\argmax} \frac{1}{U} \frac{1}{T} \sum_{u \in \mathcal{U}} \sum_{t \in \mathcal{T}} \sum_{v \in \mathcal{V}} \Delta D_{u}^{v,t}(\mathcal{S})
 \label{eq:maximization}
\end{equation}
\begin{equation}
 \mbox{s.t. } c_n(\mathcal{S})  \leq C_n, \; \forall n \in \mathcal{N} 
 \label{eq:cache_constraint_set_fn}
\end{equation}
\begin{equation}
r_n^t(\mathcal{S}) \leq R_n, \; \forall n \in \mathcal{N}\cup\{0\} , \; t\in \mathcal{T}
\label{eq:rate_constraint_set_fn}
\end{equation}
\begin{equation}
f_n^{v,t}(\mathcal{S}) \leq 1,\; \forall n \in \mathcal{N}\cup\{0\}, \; \forall v \in \mathcal{V}_p \backslash \{v_1,v_{V_p}\}, \; \forall t \in \mathcal{T}.
\label{eq:counter_constraint}
\end{equation}

\noindent
Constraints \eqref{eq:cache_constraint_set_fn} and \eqref{eq:rate_constraint_set_fn} are the cache capacity and the transmission capacity constraints. Constraint \eqref{eq:counter_constraint} guarantees that each segment is placed in the cache of a BS only once. This constraint is necessary since neither the cache cost function nor the rate cost function can distinguish between two elements $e^{v,t}_{n,\mathcal{A}_n}$ and $e^{v,t}_{n,\mathcal{A}^\prime_n}$ associated with the same segment $B^{v,t}$. In other words, for two elements $e_{n,\mathcal{A}_n}^{v,t}$, $e_{n,\mathcal{A}_n^\prime}^{v,t} \in\mathcal{S}$, the required cache space calculated by the cache cost function is $2b^{t}$, and the required rate calculated by the rate cost function is $(|\mathcal{A}_n|+|\mathcal{A}^\prime_n|)r$. In practice, however, the two elements $e_{n,\mathcal{A}_n}^{v,t}$ and $e_{n,\mathcal{A}_n^\prime}^{v,t}$ can be replaced with an equivalent element $e_{n,{\mathcal{A}_n \cup \mathcal{A}^\prime_n}}^{v,t}$. Hence, the actual cache space needed is $b^{t}$, and the actual rate needed is $|\mathcal{A}_n\cup\mathcal{A}^\prime_n|r$. Constraint \eqref{eq:counter_constraint} ensures that only a unique element $e^{v,t}_{n,\mathcal{A}_n}$ for segment $B^{v,t}$  and cache $n$ will be included in the solution set.

In the following subsection, we show that the objective function is a monotone non-decreasing submodular set function. We then leverage this property to propose computationally efficient algorithms.

\subsection{Proof of Submodularity}
\label{sec:submodularity}

Submodularity is an important property of set functions that permits to deploy greedy solutions with a good performance-complexity trade-off \cite{KrauseBOOK2013}. In this subsection, we prove that the objective function in the maximization problem in Eq.~\eqref{eq:maximization} is a monotone non-decreasing submodular function. The definition of a monotone non-decreasing set function is given in Appendix \ref{app:submod_fn}.
 
 \begin{proposition} 
The objective function in \eqref{eq:maximization} is a monotone non-decreasing set function over the ground set $\mathcal{E}$. 
\end{proposition}

\begin{proof}
 Let us consider sets $\mathcal{S}_1 \subseteq \mathcal{S}_2 \subseteq \mathcal{E}$. Monotonicity follows immediately from the observation that the distortion of a reconstructed segment $B^{v,t}$ at user $u$ can only reduce with the delivery of additional segments $B^{v^\prime,t}$; that is, the distortion reduction can only increase with the delivery of additional data segments. It thus holds that $\Delta D_u^{v,t}(\mathcal{S}_1) \leq \Delta D_u^{v,t}(\mathcal{S}_2)$, {\em i.e.,} the distortion reduction at user $u$ for reconstructing the segment $B^{v,t}$ is a monotone non-decreasing function. Hence, the objective function in \eqref{eq:maximization} is also monotone non-decreasing as a linear combination of monotone non-decreasing functions with non-negative weights. 
 \end{proof}
We will now prove that the objective function in \eqref{eq:maximization} is submodular. The following lemma will be useful in the proof.

 \begin{lemma}
Let $\mathcal{V}_1, \mathcal{V}_2$ satisfy $\mathcal{V}_p \supseteq \mathcal{V}_2 \supseteq \mathcal{V}_1 \supseteq \{v_1, v_{V_p}\} $. Consider a view $v\in\mathcal{V}_p\cup \mathcal{V}_s$. Then, for any $\tilde{v} \in \mathcal{V}_p$, we have
\begin{equation} 
\Delta D^v(\mathcal{V}_1 \cup \tilde{v}) - \Delta D^v(\mathcal{V}_1) \geq \Delta D^v(\mathcal{V}_2 \cup \tilde{v}) - \Delta D^v(\mathcal{V}_2),
\label{eq:submodularity}
\end{equation} 
where the distortion reduction function $\Delta D^v(\hat{\mathcal{V}}): 2^{\mathcal{V}_p }\rightarrow \mathbb{R}$ is defined as 
$$\Delta D^v(\hat{\mathcal{V}}) \triangleq D_{max} - d_v(v_{l}, v_{r}), $$
with $v_{l} \leq v$ and $v_r \geq v$, respectively, being the closest to $v$ left and right anchor views in $\hat{\mathcal{V}}$.
\label{lem:1}
\end{lemma}

\begin{proof}  We prove the lemma for the case where $\tilde{v} \leq v$. Due to symmetry, the same arguments hold for the case $\tilde{v} \geq v$.

For $j=1,2,$ let $v_{l_j}\in \mathcal{V}_j$ denote the left anchor view closest to $v$, such that $ 0 \leq v-v_{l_j} < v-v_{l_j^\prime}$, $\forall v_{l_j^\prime} \in \mathcal{V}_j$ with $l_j^\prime \neq l_j$. Similarly, let $v_{r_j} \in \mathcal{V}_j$ be the right anchor view closest to $v$, such that $ 0 \leq v_{r_j} -v  < v_{r_j^\prime}- v $, $\forall v_{r_j^\prime} \in \mathcal{V}_j$ with $r_j^\prime \neq r_j$. Since $\mathcal{V}_1 \subseteq \mathcal{V}_2$, we have $v_{l_1} \leq v_{l_2}$ and  $v_{r_1} \geq v_{r_2}$.  We can distinguish three cases depending on the relative position of view $\tilde{v}$ with respect to views $v_{l_1}$ and $v_{l_2}$: \textit{(i)} $\tilde{v} \leq v_{l_1}$, \textit{(ii)} $  v_{l_1} < \tilde{v} \leq v_{l_2}$, and  \textit{(iii)} $  v_{l_2} < \tilde{v} $.  We now prove for each of these three cases that the inequality in \eqref{eq:submodularity} holds.

  \renewcommand{\labelenumi}{\textit{\arabic{enumi})}}
 \renewcommand{\labelenumii}{\textit{\alph{enumii})}}
\begin{enumerate}[wide, labelwidth=!, labelindent=0pt]

\item $\tilde{v} \leq v_{l_1}$:  In this case, the addition of view $\tilde{v}$ to either of sets $\mathcal{V}_1$, $\mathcal{V}_2$ does not provide any further distortion reduction since view $\tilde{v}$ is farther from view $v$ than views $v_{l_1} $ and $v_{l_2}$; and thus, views $v_{l_1}$, $v_{l_2}$ remain as the left anchor views closest to $v$. In particular, $\Delta D^v(\mathcal{V}_1\cup\tilde{v}) = \Delta D^v(\mathcal{V}_1) =  D_{max} - d_v(v_{l_1}, v_{r_1})$, and $\Delta D^v(\mathcal{V}_2\cup\tilde{v}) = \Delta D^v(\mathcal{V}_2) =  D_{max} - d_v(v_{l_2}, v_{r_2}) $. Therefore, $\Delta D^v(\mathcal{V}_1\cup\tilde{v}) - \Delta D^v(\mathcal{V}_1)  = \Delta D^v(\mathcal{V}_2\cup\tilde{v}) - \Delta D^v(\mathcal{V}_2) = 0$.  

\item $   v_{l_1} < \tilde{v} \leq v_{l_2}$: As in the previous case, the addition of $\tilde{v}$ to set $\mathcal{V}_2$ does not provide any distortion reduction since view $\tilde{v}$ is farther from view $v$ than view $v_{l_2}$, and $v_{l_2}$ remains the left anchor view closest to $v$ in set $\mathcal{V}_2\cup \tilde{v}$. Thus, we have $\Delta D^v(\mathcal{V}_2\cup\tilde{v}) = \Delta D^v(\mathcal{V}_2) =  D_{max} - d_v(v_{l_2}, v_{r_2}) $ and $\Delta D^v(\mathcal{V}_2\cup\tilde{v}) - \Delta D^v(\mathcal{V}_2) = 0$. On the contrary, the addition of view $\tilde{v}$ to set $\mathcal{V}_1$ reduces the distortion for view $v$, since view $\tilde{v}$ is closer to view $v$ than $v_{l_1}$, {\em i.e.}, $d_v(v_{l_1}, v_{r_1}) \geq d_v(\tilde{v}, v_{r_1})$. Therefore, $\Delta D^v(\mathcal{V}_1\cup\tilde{v}) - \Delta D^v(\mathcal{V}_1) =  d_v(v_{l_1}, v_{r_1})  - d_v(\tilde{v} , v_{r_1})\geq 0$, and $\Delta D^v(\mathcal{V}_1\cup\tilde{v}) - \Delta D^v(\mathcal{V}_1)  \geq \Delta D^v(\mathcal{V}_2\cup\tilde{v}) - \Delta D^v(\mathcal{V}_2)$.

\item $   v_{l_1}\leq v_{l_2} < \tilde{v} $: In this case, view $\tilde{v}$ becomes the left anchor view closest to $v$. Thus, we have $\Delta D^v(\mathcal{V}_1\cup\tilde{v}) - \Delta D^v(\mathcal{V}_1)  \geq 0$ and $\Delta D^v(\mathcal{V}_2\cup\tilde{v}) - \Delta D^v(\mathcal{V}_2) \geq 0$. However, it is no longer possible to deduce straightforwardly which of the two gains in distortion reduction is larger, and an inspection of all the sub-cases concerning the relative positions of the views $v_{l_1}$, $v_{l_2}$, $\tilde{v}$, $v_{r_1}$, and $v_{r_2}$ with respect to view $v$ is needed. We can distinguish the following ten subcases: 
\end{enumerate}
\begin{subequations}
\allowdisplaybreaks
\begin{align}
v - \tilde{v} < v-v_{l_2} \leq v - v_{l_1} \leq v_{r_2} - v  \leq v_{r_1} - v \label{eq:case1}\\ \allowbreak
 v - \tilde{v} < v-v_{l_2} \leq v_{r_2} - v \leq v - v_{l_1}  \leq v_{r_1} - v \label{eq:case2}  \\ \allowbreak
 v - \tilde{v} \leq v_{r_2} - v \leq v-v_{l_2}  \leq v - v_{l_1}  \leq v_{r_1} - v  \label{eq:case3} \\
 v_{r_2} - v \leq v - \tilde{v} <  v-v_{l_2}  \leq v - v_{l_1}  \leq v_{r_1} - v  \label{eq:case4} \\
 v - \tilde{v} <v-v_{l_2}  \leq v_{r_2} - v \leq v_{r_1} - v \leq    v - v_{l_1}  \label{eq:case5} \\
 v - \tilde{v}  \leq v_{r_2} - v \leq v-v_{l_2}  \leq v_{r_1} - v \leq    v - v_{l_1}\label{eq:case6} \\
 v_{r_2} - v \leq v - \tilde{v} <  v-v_{l_2}  \leq v_{r_1} - v \leq    v - v_{l_1}  \label{eq:case7} \\
 v - \tilde{v} \leq v_{r_2} - v \leq v_{r_1} - v \leq  v-v_{l_2}   \leq    v - v_{l_1}  \label{eq:case8} \\   
 v_{r_2} - v \leq v - \tilde{v} \leq  v_{r_1} - v \leq  v-v_{l_2}   \leq    v - v_{l_1}  \label{eq:case9}\\  
 v_{r_2} - v \leq  v_{r_1} - v  \leq v - \tilde{v} <   v-v_{l_2}   \leq    v - v_{l_1}   \label{eq:case10}
\end{align}
\end{subequations}

\noindent
Here we show analytically that the inequality in \eqref{eq:submodularity} holds for the  case in Eq.~\eqref{eq:case1}, and provide the guidelines for showing its validity for the remaining cases, omitting the details due to limited space. From the distortion model in Eq.~\eqref{eq:dist_model} and the inequalities in \eqref{eq:case1}, we have:

\begin{align}
\allowdisplaybreaks
&\frac{\Delta D^v(\mathcal{V}_2) - \Delta D^v(\mathcal{V}_1) }{\Delta D^v(\mathcal{V}_2\cup \tilde{v}) - \Delta D^v(\mathcal{V}_1\cup \tilde{v})} \nonumber \\  
& = \frac{\mbox{e}^{a_v(v_{r_1} - v_{l_1})}( \mbox{e}^{\beta_v(v -v_{l_1})}-1) - \mbox{e}^{a_v(v_{r_2} - v_{l_2})}( \mbox{e}^{\beta_v(v -v_{l_2})}-1) }{\mbox{e}^{a_v(v_{r_1} - \tilde{v})}( \mbox{e}^{\beta_v(v -\tilde{v})}-1) - \mbox{e}^{a_v(v_{r_2} - \tilde{v})}( \mbox{e}^{\beta_v(v -\tilde{v})}-1)} \nonumber \\ 
 &\overset{v-v_{l_1} > v-v_{l_2}}{\geq}  \frac{(\mbox{e}^{a_v(v_{r_1} - v_{l_1})} - \mbox{e}^{a_v(v_{r_2} - v_{l_2})})( \mbox{e}^{\beta_v(v -v_{l_2})}-1) }{(\mbox{e}^{a_v(v_{r_1} - \tilde{v})}-\mbox{e}^{a_v(v_{r_2} - \tilde{v})})( \mbox{e}^{\beta_v(v -\tilde{v})}-1)} \nonumber\\ 
& \overset{v-v_{l_2} > v-\tilde{v}}{\geq} \frac{\mbox{e}^{a_v(v_{r_1} - v_{l_1})} - \mbox{e}^{a_v(v_{r_2} - v_{l_2})}}{\mbox{e}^{a_v(v_{r_1} - \tilde{v})}-\mbox{e}^{a_v(v_{r_2} - \tilde{v})}} \label{eq:submod_proof} \\ 
 &= \frac{(\mbox{e}^{a_v(v_{r_1} - v_{r_2} + v_{l_2}-v_{l_1})}-1)\mbox{e}^{a_v(v_{r_2-v_{l_1}})}}{(\mbox{e}^{a_v(v_{r_1}-v_{r_2})}-1)  \mbox{e}^{a_v(v_{r_2}-\tilde{v})} } \nonumber \\
&  \overset{v_{r_2} - v_{l_1} \geq v_{r_2}-\tilde{v}}{\geq} \frac{\mbox{e}^{a_v(v_{r_1} - v_{r_2} + v_{l_2}-v_{l_1})}-1}{\mbox{e}^{a_v(v_{r_1}-v_{r_2})}-1} \overset{v_{l_2} \geq v_{l_1}}{\geq} 1. \nonumber
\end{align}

\noindent
Inequality \eqref{eq:submodularity} follows immediately from \eqref{eq:submod_proof}. Using the same procedure, we can prove that \eqref{eq:submodularity} holds for the cases in \eqref{eq:case2}, \eqref{eq:case3} \eqref{eq:case5}, \eqref{eq:case6}, \eqref{eq:case8}. For the rest of the cases, we form the expression 
$$\frac{\Delta D^v(\mathcal{V}_1 \cup \tilde{v}) - \Delta D^v(\mathcal{V}_1) }{\Delta D^v(\mathcal{V}_2\cup \tilde{v}) - \Delta D^v(\mathcal{V}_2\cup \tilde{v})}, $$
and using similar arguments as before, we prove that this expression is greater or equal to 1.

\end{proof}

Intuitively, the above result can be explained by the fact that the quality of the left reference anchor view improves more when adding view $\tilde{v}$ to set $\mathcal{V}_1$ than when adding $\tilde{v}$ to set $\mathcal{V}_2$, since $\tilde{v} - v_{l_1} \geq \tilde{v} - v_{l_2}$. Thus, the gain in the distortion reduction is higher when adding view $\tilde{v}$ to set $\mathcal{V}_1$ compared to adding it to $\mathcal{V}_2$.

\begin{proposition}
The objective function in \eqref{eq:maximization} is a submodular set function over the ground set $\mathcal{E}$. 
\end{proposition}

\begin{proof}
Since a non-negative linear combination of monotone submodular functions is also submodular \cite{KrauseBOOK2013}, it is sufficient to show that the distortion reduction $\Delta D^{v,t}_{u}(\mathcal{S}): 2^\mathcal{E} \to \mathbb{R}$ is monotone submodular $\forall u \in \mathcal{U}$, $\forall v \in \mathcal{V}_p \cup \mathcal{V}_s$, $\forall t \in \mathcal{T}$. Let us consider the sets $\mathcal{S}_1 \subseteq \mathcal{S}_2 \subseteq \mathcal{E}$, and an element $e_{n,\mathcal{A}_n}^{\tilde{v},\tilde{t}} \in \mathcal{E}\backslash \mathcal{S}_2$. This element represents the joint placement of segment $B^{\tilde{v}, \tilde{t}}$ in the cache of BS $n$, and its delivery from BS $n$ to the set of users $\mathcal{A}_n$. If $u \notin \mathcal{A}_n$ or $\tilde{t} \neq t$, then adding $e_{n,\mathcal{A}_n}^{\tilde{v},\tilde{t}}$ to sets $\mathcal{S}_1$ and  $\mathcal{S}_2$ does not affect the distortion reduction at user $u$ for segment $B^{v,t}$ since user $u$ does not receive any additional segments with respect to those received according to the joint caching and scheduling policies  defined by sets $\mathcal{S}_1$ and $\mathcal{S}_2$. Thus, $\Delta D_u^{v,t}(\mathcal{S}_1\cup e_{n,\mathcal{A}_n}^{\tilde{v},\tilde{t}}) - \Delta D_u^{v,t}(\mathcal{S}_1) = \Delta D_u^{v,t}(\mathcal{S}_2\cup e_{n,\mathcal{A}_n}^{\tilde{v},\tilde{t}}) - \Delta D_u^{v,t}(\mathcal{S}_2) = 0$. Next, we focus on the case $u \in \mathcal{A}_n$ and $\tilde{t} = t$, {\em i.e.}, user $u$ belongs to the group of users to which segment $B^{\tilde{v}, t}$ is delivered. We associate set $\mathcal{S}_j$, $j = 1,2$, with set $\mathcal{V}_j\subseteq \mathcal{V}_p$, where $v \in \mathcal{V}_j$ iff $e^{v,t}_{n,\mathcal{A}_n} \in \mathcal{S}_j$. From the definition of sets $\mathcal{V}_j$, and since $\mathcal{S}_1 \subseteq \mathcal{S}_2$, it holds that $\mathcal{V}_1 \subseteq \mathcal{V}_2$. Due to the assumption that all the segments of the leftmost and rightmost views are delivered to the users, we also have $\mathcal{V}_j \supseteq \{v_1, v_{V_p}\}$. From Lemma \ref{lem:1} it follows that  $\Delta D^v(\mathcal{V}_1 \cup \tilde{v}) - \Delta D^v(\mathcal{V}_1) \geq \Delta D^v(\mathcal{V}_2 \cup \tilde{v}) - \Delta D^v(\mathcal{V}_2)$; and therefore, $\Delta D^{v,t}_u(\mathcal{S}_1 \cup e^{\tilde{v},t}_{n,\mathcal{A}_n}) - \Delta D^{v,t}_u(\mathcal{S}_1) \geq \Delta D^{v,t}_u(\mathcal{S}_2 \cup  e^{\tilde{v},t}_{n,\mathcal{A}_n}) - \Delta D^{v,t}_u(\mathcal{S}_2)$,  which completes the proof.
\end{proof}

\subsection{Greedy algorithms}
\label{sec:greedy}

In the previous section, we have shown that the objective function in the maximization problem in \eqref{eq:maximization} is a monotone non-decreasing submodular function. We can now show that the optimization problem defined  by Eqs.~\eqref{eq:maximization}-\eqref{eq:counter_constraint} is in the form of a submodular set function maximization problem subject to a separable $d$-dimensional knapsack constraint defined in Appendix \ref{app:submod_fn}. By inspection of the constraints in \eqref{eq:cache_constraint_set_fn}-\eqref{eq:counter_constraint}, it is straightforward to see that they can be partitioned into $d^\prime = 3$ disjoint sets of constrains with $M_1 = N$ cache constraints in the first set, $M_2 = (N+1)T$ rate constraints in the second set, and $M_3 = (N+1)(V_p-2)T$ constraints in the third set that ensure the uniqueness of the selected elements. We can therefore apply the uniform cost (UC) and the weighted cost-benefit (WCB) greedy algorithms described in Appendix \ref{app:greedy} to efficiently solve the maximization problem in \eqref{eq:maximization}-\eqref{eq:counter_constraint}. For the sake of completeness, we summarize the UC and WCB algorithms as applied to the maximization problem in \eqref{eq:maximization}-\eqref{eq:counter_constraint} in Algorithms \ref{algo:uniform_cost} and \ref{algo:weighted_cost}, respectively. According to Theorem \ref{th:approx_ratio} provided in Appendix \ref{app:approx_ratio}, at least one of the two greedy algorithms achieves the approximation ratio of $\frac{1}{2}(1-\mbox{e}^{-1})$. It is worth noting that the WCB greedy algorithm was used in \cite{LiTCSVT2017} to maximize a submodular objective function subject to two knapsack constraints. However, the authors did not provide any theoretical guarantees on its performance. To the best of our knowledge, our work presents the first constant approximation ratio for solving the submodular set function maximization problem subject to a $d$-dimensional knapsack constraint by means of greedy algorithms. 

\begin{algorithm}[t]
\caption{Uniform cost greedy algorithm}
\label{algo:uniform_cost}
\begin{algorithmic}[1]
\STATE \textbf{Input:} $\mathcal{E}$, value query oracle $\Delta D(\mathcal{S})$, $C_n$, $R_n$, cost functions $c_n$, $r_n^t$, $f_n^{v,t}$
\STATE \textbf{Initialization:} $\mathcal{S}_{UC} \leftarrow \emptyset$, $k \leftarrow 0$
\WHILE {$\mathcal{E} \backslash \mathcal{S}_{UC} \neq \emptyset $ }
\STATE $k \leftarrow k+1$
\STATE $e_k \leftarrow \underset{e^{v,t}_{n,\mathcal{A}_n}\in \mathcal{E} \backslash \mathcal{S}_{UC}}{\argmax} \Delta D(\mathcal{S}_{UC}\cup e^{v,t}_{n,\mathcal{A}_n}) - \Delta D(\mathcal{S}_{UC}) $ 
\IF {$c_n(\mathcal{S}_{UC}\cup e_k)\leq C_n$, $r_n^t(\mathcal{S}_{UC}\cup  e_k) \leq R_n$,  $f_n^{v,t}(\mathcal{S}_{UC}\cup  e_k) \leq 1$} 
\STATE $\mathcal{S}_{UC} \leftarrow \mathcal{S}_{UC} \cup e_k$, 
\ELSE 
\STATE $ \mathcal{E} \leftarrow \mathcal{E} \backslash e_k$
\ENDIF
\ENDWHILE
\STATE \textbf{Output:} $\mathcal{S}_{UC} $
\end{algorithmic}
\end{algorithm}

\begin{algorithm}[h]
\caption{Weighted cost-benefit greedy algorithm}
\label{algo:weighted_cost}
\begin{algorithmic}[1]
\STATE \textbf{Input:} $\mathcal{E}$, value query oracle $\Delta D(\mathcal{S})$, $C_n$, $R_n$, cost functions $c_n$, $r_n^t$, $f_n^{v,t}$, weights $\lambda_1, \; \lambda_2, \; \lambda_3$ 
\STATE \textbf{Initialization:} $\mathcal{S}_{WCB} \leftarrow \emptyset$, $k \leftarrow 0$
\WHILE {$\mathcal{E}  \backslash \mathcal{S}_{WCB}\neq \emptyset $ }
\STATE $k \leftarrow k+1$
\STATE
\begin{eqnarray}
e_k \leftarrow \underset{e^{v,t}_{n,\mathcal{A}_n}\in \mathcal{E} \backslash \mathcal{S}_{WCB}}{\argmax} &\lambda_1\frac{\Delta D(\mathcal{S}_{WCB}\cup e^{v,t}_{n,\mathcal{A}_n}) - \Delta D (\mathcal{S}_{WCB})} {\sum_{n^\prime}c_{n^\prime}(e^{v,t}_{n,\mathcal{A}_n})} \nonumber\\
+&\lambda_2\frac{\Delta D (\mathcal{S}_{k-1}\cup e^{v,t}_{n,\mathcal{A}_n}) - \Delta D(\mathcal{S}_{k-1})} {\sum_{n^\prime} \sum_{t^\prime}r_{n^\prime}^{t^\prime}(e^{v,t}_{n,\mathcal{A}_n})}\nonumber\\
+&\lambda_3 \frac{\Delta D (\mathcal{S}_{k-1}\cup e^{v,t}_{n,\mathcal{A}_n}) - \Delta D (\mathcal{S}_{k-1})}{\sum_{n^\prime} \sum_{v^\prime} \sum_{t^\prime}f_{n^\prime}^{v^\prime,t^\prime}(e^{v,t}_{n,\mathcal{A}_n})} \nonumber
 \end{eqnarray}
 \IF {$c_n(\mathcal{S}_{WCB}\cup e_k)\leq C_n$, $r_n^t(\mathcal{S}_{WCB} \cup e_k) \leq R_n$,  $f_n^{v,t}(\mathcal{S}_{WCB} \cup e_k) \leq 1$ } 
\STATE $\mathcal{S}_{WCB} \leftarrow \mathcal{S}_{WCB} \cup e_k$,
\ELSE
\STATE $ \mathcal{E} \leftarrow \mathcal{E} \backslash e_k$
\ENDIF
\ENDWHILE
\STATE \textbf{Output:} $\mathcal{S}_{WCB} $
\end{algorithmic}
\end{algorithm}

\section{Performance evaluations}
\label{sec:evaluation}

For the performance evaluations, we consider a circular cell with the MBS located at its centre. The transmission range of the MBS is set to 400m. A total number of 20 SBSs, each with a coverage radius of 100m, are placed uniformly at random over the cell. The transmission capacity of the SBSs is set to 100Mbps. We consider 200 wireless users uniformly distributed across the macro cell. 

The IMVS system consists of $V_p = 8$ cameras. Each camera generates a video stream encoded at $r= 2$ Mbps and divided into $T = 20$ segments of equal size. We assume that $L = 3$ virtual viewpoints can be synthesized between any two adjacent anchor views. The distortion of the synthesized views is computed based on the model in \eqref{eq:dist_model}.  We assume that the users select the first segment among the captured views uniformly at random. Then, during the streaming session, each user can switch from view $v_i$ to a neighbouring anchor or virtual view $v_j$ with probability $p(v_j|v_i) \propto \frac{1}{\sqrt{2\pi}\sigma}e^{-\frac{(v_j - v_i)^2}{2\sigma^2}}$ for $ |v_j-v_i| \leq W$, and 0 otherwise. For our evaluations, we set $W = 8$ and $\sigma^2  = 5/(L+1)$. From this model, we calculate the popularity distribution $p^{v,t}$ of the video segments. 

We compare the proposed greedy joint caching and scheduling algorithms with a maximum popularity algorithm. The latter fills each SBS's cache with the most popular video segments. It then performs greedy scheduling independently of the cache placement phase. We evaluate both UC greedy and WCB greedy scheduling for the maximum popularity algorithm.

\begin{figure}[t]
	\begin{center}
		\includegraphics[width = 0.45 \textwidth]{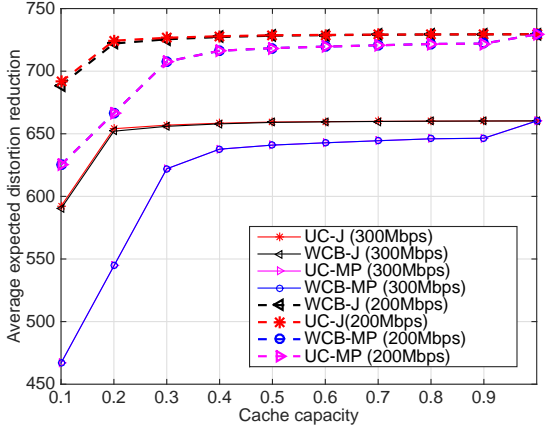} 
	\end{center}
		\caption{Average expected distortion reduction vs the cache capacity at the SBSs, expressed as the percentage of the total size of the multiview video.}
	\label{fig:dist_vs_cache}
		\vspace{-0.4cm}
\end{figure}

Fig.~\ref{fig:dist_vs_cache} shows the average expected distortion reduction versus the cache capacity of the SBSs expressed as a percentage of the total size of the multiview video. UC-J and WCB-J denote the UC and WCB greedy algorithms, respectively, for joint caching and scheduling. UC-MP and WCB-MP denote the maximum popularity caching algorithm with UC and WCB greedy scheduling, respectively. We present results for a total transmission capacity of $200$Mbps and $300$Mbps for the MBS. For the WCB algorithm we have used $\lambda_1 = 0.2$, $\lambda_2 = 0.5$ and $\lambda_3 = 0.3$. The results indicate that the joint caching and scheduling algorithms outperform the maximum popularity counterparts for all values of the cache capacity. For low values of the cache capacity, the improvement in the performance is significant as the maximum popularity algorithm caches the same content in all SBSs; thus the content diversity across the network is limited. Along with the most popular content cached only in few SBSs, the joint caching and scheduling algorithm also caches the less popular content, which, when delivered to the users, improves the reconstruction quality of the views. It is worth noting that this range of capacity values is of great practical interest as SBSs are typically assumed to cache only 5-10\% of the total video catalogue \cite{MaggiARXIV2012,ParisisCOMCOM2017}. The performance of all the algorithms becomes limited by the insufficient transmission capacity of the network. Thus, even though all the SBSs can cache almost all of the contents, they cannot be delivered to the users. 

\begin{figure}[t]
	\begin{center}
		\includegraphics[width = 0.45 \textwidth]{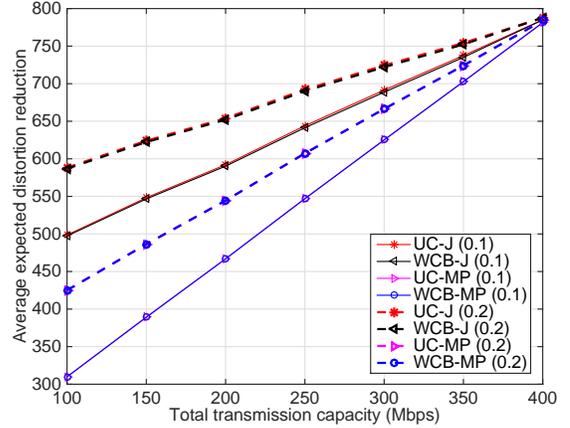} 
	\end{center}
		\caption{Average expected distortion reduction vs the total transmission rate of the MBS.}
	\label{fig:dist_vs_rate}
	\vspace{-0.4cm}
\end{figure}

In Fig.~\ref{fig:dist_vs_rate} we show the average expected distortion reduction versus the total transmission capacity of the MBS for cache capacity equal to 10\% and 20\% of the total size of the video. As the total transmission capacity of the MBS increases, the average expected quality of the multiview video delivered to the users improves. We can see that the joint caching and scheduling algorithm outperforms the maximum popularity algorithm for all values of the MBS transmission capacity. Although the cache capacity of the SBSs is limited, our algorithm performs much better compared to the maximum popularity algorithm due to the more efficient use of the available cache and transmission capacities. As previously, the content diversity is higher when the caching and scheduling policies are optimized jointly. It is worth noting that to achieve the same average expected distortion reduction, the maximum popularity caching algorithm requires a much higher transmission rate to be allocated by the MBS compared to the case of joint cache and scheduling optimization.

Finally, we note that the UC and WCB greedy algorithms in Figs.~\ref{fig:dist_vs_cache} and \ref{fig:dist_vs_rate} perform identically. This is due to the fact that all video segments have the same size. We expect that in the case of multiple multiview videos encoded at different rates, or video segments of unequal duration, the performance of the two algorithms would be different. We leave this investigation for our future work. 

\section{Conclusions and future work}
\label{sec:conclusions}

We have presented a framework for jointly optimizing the caching and scheduling policy for interactive multiview video delivery over a wireless cellular network. Unlike existing works for wireless edge caching, our scheme takes into account the quality of the video delivered to the users and the rate requirements for real-time video delivery. Numerical evaluation of our scheme shows that the joint policy performs significantly better than the independent caching and scheduling policies for the case of multiview video. In our future work, we will investigate ways to simplify the expression for calculating the distortion of the delivered video, with the aim of obtaining a convex problem which can be solved for optimality. A possible approach is to organize the views into embedded sets that progressively improve the quality of the delivered video. This will also relax the constraint of caching the whole segment in the same SBS, and will permit to cache parts of the same video segment encoded with a rateless code in different SBSs. The latter will allow to transform the integer optimization problem into a linear one. 

\appendices
 
\section{Submodular functions}
\label{app:submod_fn}

\subsection{Definitions and properties}
\label{app:def_prop}

Here we recall some basic definitions and results from the theory of submodular functions.

\begin{definition} \cite{KrauseBOOK2013}
A function $g: 2^\mathcal{W} \rightarrow \mathbb{R}$ defined over a ground set $\mathcal{W}$ is \textit{submodular}, if for every $\mathcal{Z}_1 \subseteq \mathcal{Z}_2 \subseteq \mathcal{W}$ and $w \in \mathcal{W}\backslash \mathcal{Z}_2$, 
\begin{equation} \small
g(\mathcal{Z}_1\cup w) - g(\mathcal{Z}_1) \geq g(\mathcal{Z}_2 \cup w) - g(\mathcal{Z}_2).
\label{eq:submod1}
\end{equation}

Alternatively, $g$ is submodular, if for every $\mathcal{Z}_1 , \mathcal{Z}_2 \subseteq \mathcal{W}$, 
\begin{equation} \small
g(\mathcal{Z}_1\cap \mathcal{Z}_2) + g(\mathcal{Z}_1\cup \mathcal{Z}_2) \leq g(\mathcal{Z}_1) + g(\mathcal{Z}_2).
\label{eq:submod2}
\end{equation}
\end{definition}

\begin{definition} \cite{KrauseBOOK2013}
A function $g: 2^\mathcal{W} \rightarrow \mathbb{R}$ defined over a ground set $\mathcal{W}$ is monotone non-decreasing if for every $\mathcal{Z}_1 \subseteq \mathcal{Z}_2 \subseteq \mathcal{W}$, $g(\mathcal{Z}_1) \leq g(\mathcal{Z}_2)$. 
\end{definition}

\begin{proposition} \cite{NemhauserMATHPROG1978} If $g: 2^\mathcal{W} \rightarrow \mathbb{R}$ is a monotone non-decreasing submodular function defined over the ground set $\mathcal{W}$, then
\begin{equation} \small
g(\mathcal{Z}_1) \leq g(\mathcal{Z}_2) + \sum_{w \in \mathcal{Z}_1\backslash \mathcal{Z}_2} g(\mathcal{Z}_2\cup w ) - g(\mathcal{Z}_2)
\end{equation}
for all $\mathcal{Z}_1, \mathcal{Z}_2 \subseteq \mathcal{W} $
\label{eq:prop1}
\label{prop1}
\end{proposition}


\subsection{Submodular function maximization with a $d$-dimensional knapsack constraint}
Let $g: 2^\mathcal{W} \rightarrow \mathbb{R}$ be a monotone non-decreasing submodular function defined over the ground set $\mathcal{W}$. The submodular function maximization problem subject to a $d$-dimensional knapsack constraint is formulated as 
\begin{equation} \small
\begin{split}
\max_{\mathcal{Z}\in\mathcal{W}} & \;g(\mathcal{Z}) \\
\mbox{s.t. } \quad h_i(\mathcal{Z}) \leq  & H_i, \quad i = 1, 2, \dots, d,
\end{split}
\label{eq:submodmax}
\end{equation}
\noindent where $d$ is the number of knapsack dimensions, $h_i(\mathcal{Z}) = \sum_{z\in\mathcal{Z}}h_i(z)$ is the cost function of the $i$th dimension, and $h_i(z) > 0$. Each dimension of the knapsack can be viewed as a resource, with $H_i$ being the total budget of the $i$th resource.

We now consider a special case of the maximization problem in \eqref{eq:submodmax}. Let us assume that the $d$ knapsack constraints in \eqref{eq:submodmax} can be partitioned into $d^\prime$ disjoint subsets such that, for every subset  $i = 1, 2, \dots, d^\prime$ of the constraints, the ground set $\mathcal{W}$ can be partitioned into $M_i$ disjoint subsets $\mathcal{W}_i^1, \mathcal{W}_i^2 \dots, \mathcal{W}_i^{M_i}$, where $M_i$ is the number of constraints in the $i$th subset and $\sum_{i=1}^{d^\prime}M_i = d$. Let  $h_i^m(\mathcal{Z}) = \sum_{z\in\mathcal{Z}}h_i^m(z)$ be the cost function for the $m$th constraint in the $i$th subset of constraints with $m = 1,2,\dots, M_i$. We further assume that $ h_i^m(z) > 0$ if $z \in \mathcal{W}_i^m$, and $0$ if $z \notin \mathcal{W}_i^m$, only a subset of elements in $ \mathcal{W}$ consume a non-zero amount of the $m$th resource in the $i$th set of resources, and there are $M_i$ such disjoint subsets of elements associated with the $i$th set of resources for $i = 1, 2, \dots, d^\prime$.  Formally, this problem can be written as 
\begin{equation} \small
\begin{split}
\max_{\mathcal{Z}\in\mathcal{W}} & \;g(\mathcal{Z}) \\
\mbox{s.t. } \quad h_i^m(\mathcal{Z}) \leq   H_i^m, \quad & i = 1, 2, \dots, d^\prime,  \quad m = 1, 2,\dots M_i,
\end{split}
\label{eq:submodmax_disjoint}
\end{equation}
where  $h_i^m(\mathcal{Z}) = \sum_{z\in\mathcal{Z}}h_i^m(z) = \sum_{z\in\mathcal{Z}\cap \mathcal{W}_i^m}h_i^m(z)$, and $H_i^m$ is the total budget for the $m$th resource in the $i$th set of constraints. We refer to this problem as the submodular function maximization with a separable $d$-dimensional knapsack constraint.


\section{Greedy algorithms for submodular function maximization}
\label{app:greedy}

\subsection{The uniform cost (UC) greedy algorithm}
\label{sec:UC}
The UC greedy algorithm \cite{LeskovecSIGKDD2007} for the submodular function maximization problem with $d$-dimensional knapsack constraint stated in \eqref{eq:submodmax}  is described in Algorithm \ref{algo:uniform_cost_appendix}. The algorithm takes as input the ground set $\mathcal{W}$, a value query oracle $g(\mathcal{Z})$ that returns the value of the objective function in \eqref{eq:submodmax} for some subset $\mathcal{Z}$ of the ground set, the cost functions  $h_i(\mathcal{Z})$ and the values of the total budgets $H_i$, $i =1,2,\dots, d$.  The algorithm starts with an empty solution set, and at the $k$-th iteration picks the element from the ground set that maximizes the gain with respect to the solution set computed at step $k-1$ (step 5 of Algorithm \ref{algo:uniform_cost_appendix}).  If this choice satisfies the $d$-dimensional knapsack constraint specified in \eqref{eq:submodmax}, the element is added to the solution set. Otherwise, the solution set is not updated and the element is removed from the ground set. This procedure is repeated until all elements from the ground set have been either included in the solution set or removed from the ground set.  When Algorithm \ref{algo:uniform_cost_appendix} is applied to the separable $d$-dimensional knapsack constraint problem in \eqref{eq:submodmax_disjoint}, the condition in line 6 must be replaced with the condition given in the parentheses. 

\begin{algorithm}[t]
\caption{UC greedy algorithm}
\label{algo:uniform_cost_appendix}
\begin{algorithmic}[1]
\STATE \textbf{Input:} $\mathcal{W}$, value query oracle $g(\mathcal{Z})$, cost functions $h_i(\mathcal{Z})$, total budget values $H_i$
\STATE \textbf{Initialization:} $\mathcal{Z}_{UC} \leftarrow \emptyset$, $jk\leftarrow 0$
\WHILE {$\mathcal{W} \backslash \mathcal{Z}_{UC} \neq \emptyset $ }
\STATE $k \leftarrow k+1$
\STATE $w_k \leftarrow \underset{w\in \mathcal{W}\backslash \mathcal{Z}_{UC}}{\argmax} g(\mathcal{Z}_{UC}\cup w) - g(\mathcal{Z}_{UC}) $ 
\IF {$h_i(\mathcal{Z}_{UC} \cup w)\leq H_i$,  $\forall i $ $\Big( h_i^m(\mathcal{Z}_{UC} \cup w)\leq H_i^m \Big)$ } 
\STATE $\mathcal{Z}_{UC} \leftarrow \mathcal{Z}_{UC} \cup w$
\ELSE 
\STATE  $ \mathcal{W}\leftarrow \mathcal{W} \backslash w_k$
\ENDIF
\ENDWHILE
\STATE \textbf{Output:} $\mathcal{Z}_{UC} $
\end{algorithmic}
\end{algorithm}

\subsection{The WCB greedy algorithm}
\label{sec:WCB}
The UC greedy algorithm presented above can perform arbitrarily poorly as it does not take into account the cost of the element selected greedily at each iteration \cite{KrauseBOOK2013}. This shortcoming is addressed in the cost-benefit greedy algorithm for the sumbodular function maximization problem with a single knapsack constraint \cite{SviridenkoOPRESLET2004}. In the cost-benefit greedy algorithm, the next element to be included in the solution set is the element that maximizes the gain-to-cost ratio. However, in a $d$-dimensional knapsack constraint problem, each element is associated with a $d$-dimensional cost vector. To account for the $d$ different costs, the authors in \cite{LiTCSVT2017} have introduced the WCB algorithm. The WCB algorithm is summarized in Algorithm \ref{algo:weighted_cost_appendix}. It works similarly to the UC greedy algorithm, but instead of selecting the element that maximizes the gain resulting from adding this element to the solution set, it maximizes a weighted sum of the gain-to-cost ratio per each dimension of the $d$-dimensional knapsack constraint.  The weights $\lambda_i$ satisfy $\sum_{i=1}^d \lambda_i =1$ and can be chosen arbitrarily to reflect the significance of the dimensions.  When applied to solving the separable $d$-dimensional knapsack constraint problem stated in \eqref{eq:submodmax_disjoint}, the assignment in line 5 and the condition in line 6 of Algorithm \ref{algo:weighted_cost_appendix} must be replaced by the corresponding assignment and condition given in the parentheses. 

\begin{algorithm}[t]
\caption{WCB greedy algorithm }
\label{algo:weighted_cost_appendix}
\begin{algorithmic}[1]
\STATE \textbf{Input:} $\mathcal{W}$, value query oracle $g(\mathcal{Z})$, cost functions $h_i(\mathcal{Z})$,  total budget values $H_i$, weights $\lambda_i$ 
\STATE \textbf{Initialization:}  $\mathcal{Z}_{WCB} \leftarrow \emptyset$, $k \leftarrow 0$
\WHILE {$\mathcal{W} \backslash \mathcal{Z}_{WCB}\neq \emptyset $ }
\STATE $k \leftarrow k+1$
\STATE 
$w_k \leftarrow \underset{w\in \mathcal{W} \backslash \mathcal{Z}_{WCB}}{\argmax} \sum_{i=1}^d \lambda_i\frac{g(\mathcal{Z}_{WCB}\cup w) - g (\mathcal{Z}_{WCB})} {h_i(w)} $ \\
$\Big(w_k \leftarrow \underset{w\in \mathcal{W} \backslash \mathcal{Z}_{WCB}}{\argmax} \sum_{i=1}^{d^\prime} \lambda_i\frac{g(\mathcal{Z}_{WCB}\cup w) - g (\mathcal{Z}_{WCB})} {\sum_{m=1}^{M_i}h_i^m(w)} \Big)$
\IF {$h_i(\mathcal{Z}_{WCB} \cup w)\leq H_i$, $\forall i$ \\  $\Big(h_i^m(\mathcal{Z}_{WCB} \cup w)\leq H_i^m$, $\forall i$, $\forall m\Big)$} 
\STATE $\mathcal{Z}_{WCB} \leftarrow \mathcal{Z}_{WCB} \cup w$
\ELSE 
\STATE  $ \mathcal{W} \leftarrow \mathcal{W} \backslash w$
\ENDIF
\ENDWHILE
\STATE \textbf{Output:} $\mathcal{Z}_{WCB} $
\end{algorithmic}
\end{algorithm}


\section{Approximation ratio}
\label{app:approx_ratio}

Here we show that at least one of the two greedy algorithms, namely the UC and the WCB algorithms,  does not perform too badly. More rigorously, let $\mathcal{Z}_{OPT}$ denote the solution of the optimization problem defined in \eqref{eq:submodmax}, {\em i.e.}, 
\begin{equation} \small
\mathcal{Z}_{OPT} \triangleq \argmax_{\mathcal{Z}\in\mathcal{W}} g(\mathcal{Z}) \mbox{ and } h_i(\mathcal{Z}_{OPT}) \leq H_i, \quad \forall i.
\end{equation}
Let $\mathcal{Z}_{UC}$ and $\mathcal{Z}_{WCB}$ be the solutions returned by the UC and WCB greedy algorithms, respectively, applied to problem \eqref{eq:submodmax}. The following theorem states that the worst case performance guarantee of the two greedy algorithms for solving the optimization problem in \eqref{eq:submodmax} is $\frac{1}{2}(1-e^{-1})$. This theorem generalizes the result in \cite{KrauseTECHREP2005} to the case of $k$-dimensional knapsack constraint. 
\begin{theorem}
Let  $\mathcal{Z}^*  \triangleq \argmax_{\mathcal{Z}\in\{  \mathcal{Z}_{UC}, \mathcal{Z}_{WCB}\}} g(\mathcal{Z})$ when applying the UC and WCB algorithms to the optimization problem in \eqref{eq:submodmax}. Then,  $g(\mathcal{Z}^*) > \frac{1}{2}(1-e^{-1}) g(\mathcal{Z}_{OPT})$.
\label{th:approx_ratio}
\end{theorem}

\begin{proof}
Let $\mathcal{Z}_i =\{w_1, w_2, \dots, w_i\}$ be the value of $\mathcal{Z}_{WCB}$  at the $k_i$th iteration of Algorithm \ref{algo:weighted_cost_appendix}, where $ i\leq k_i$.  Then,
\begin{align}\small \allowdisplaybreaks
&g(\mathcal{Z}_{OPT}) 
 \overset{(a)}{\leq} g(\mathcal{Z}_{i-1}) +\sum_{w\in \mathcal{Z}_{OPT} \backslash \mathcal{Z}_{i-1}} g(\mathcal{Z}_{i-1}\cup w) -g(\mathcal{Z}_{i-1} ) \nonumber \\ 
 &= g(\mathcal{Z}_{i-1}) + \nonumber \\ 
 & \sum_{w\in \mathcal{Z}_{OPT} \backslash \mathcal{Z}_{i-1}} \Big(\sum_{j=1}^d \frac{\lambda_j}{h_j(w)}\Big)\frac{g(\mathcal{Z}_{i-1}\cup w) -g(\mathcal{Z}_{i-1} )}{\sum_{j=1}^d \frac{\lambda_j}{h_j(w)}} \nonumber \\ 
 &\overset{(b)}{\leq} g(\mathcal{Z}_{i-1}) + \sum_{j=1}^d \frac{\lambda_j}{h_j(w_i)}\Big( g(\mathcal{Z}_{i}) -g(\mathcal{Z}_{i-1} ) \Big) \nonumber \\ 
 & \quad \quad\quad\quad \quad\quad \quad\quad \quad \quad\sum_{w\in \mathcal{Z}_{OPT} \backslash \mathcal{Z}_{i-1}} \frac{1}{\sum_{j=1}^d \frac{\lambda_j}{h_j(w)}} \label{eq:inequality1}\\
  &\overset{(c)}{< }g(\mathcal{Z}_{i-1}) + \sum_{j=1}^d \frac{\lambda_j}{h_j(w_i)}\Big( g(\mathcal{Z}_{i}) -g(\mathcal{Z}_{i-1} ) \Big) \nonumber \\
  & \quad \quad\quad \quad\quad \quad\quad \quad \quad\quad\sum_{w\in \mathcal{Z}_{OPT} \backslash \mathcal{Z}_{i-1}} \sum_{j=1}^d\frac{h_j(w)}{\lambda_j} \nonumber \\
 &\overset{(d)}{\leq}g(\mathcal{Z}_{i-1})  + \sum_{j=1}^d \frac{\lambda_j}{h_j(w_i)} \sum_{j=1}^d\frac{H_j}{\lambda_j} \Big( g(\mathcal{Z}_{i}) -g(\mathcal{Z}_{i-1} ) \Big) \nonumber
\end{align}

In the above series of inequalities, inequality (a) is due to Proposition \ref{prop1}, inequality (b) results from the greediness of Algorithm \ref{algo:weighted_cost_appendix} and inequality (c) uses the inequality 
\begin{equation} \small
\frac{1}{x_1+x_2} < \frac{1}{x_1}+\frac{1}{x_2}, \quad  \mbox{for} \quad  x_1,x_2 >0
\label{eq:inequality2}
\end{equation}
 Finally, inequality (d) results from the fact that $\sum_{w\in \mathcal{Z}_{OPT} \backslash \mathcal{Z}_{i-1}} h_j(w) \allowbreak \leq H_j$.

Subtracting  $\sum_{j=1}^d \frac{\lambda_j}{h_j(w_i)} \sum_{j=1}^d\frac{H_j}{\lambda_j} g(\mathcal{Z}_{OPT})$ from both sides of the inequality \eqref{eq:inequality1} and rearranging the terms, we obtain the following recursive inequality
\begin{equation} \small
\begin{split}
g(&\mathcal{Z}_{i})  - g(\mathcal{Z}_{OPT}) \\
& > \Big(   1 - \frac{1}{ \sum_{j=1}^d \frac{\lambda_j}{h_j(w_i)}  \sum_{j=1}^d\frac{H_j}{\lambda_j} } \Big) \Big(  g(\mathcal{Z}_{i-1}) - g(\mathcal{Z}_{OPT})   \Big)\\
&  \overset{e}{>}  \Big(1 - \frac{1} {\sum_{j=1}^d\frac{H_j}{h_j(w_i)}}  \Big) \Big(  g(\mathcal{Z}_{i-1}) - g(\mathcal{Z}_{OPT})   \Big) \\
&  \overset{f}{>} \Big(1 -\sum_{j=1}^d\frac{h_j(w_i)}{H_j} \Big)\Big(  g(\mathcal{Z}_{i-1}) - g(\mathcal{Z}_{OPT})   \Big)
\end{split}
\label{eq:inequality3}
\end{equation}
To obtain inequality (e) we have used the following inequality
\begin{equation} \small
\frac{1}{\sum_{i=1}^m a_i \sum_{i=1}^m b_i} < \frac{1}{\sum_{i=1}^m a_ib_i}, \quad \mbox{for} \quad  a_i,b_i > 0
\label{eq:inequality4}
\end{equation}
 
\noindent while (f) uses inequality \eqref{eq:inequality2}. After solving inequality \eqref{eq:inequality3} recursively, we obtain 
\begin{equation} \small
\begin{split}
g(\mathcal{Z}_i) &> \Big(1-\prod_{k=1}^i \Big(1-\sum_{j=1}^d\frac{h_j(w_k)}{H_j}\Big) \Big)g(\mathcal{Z}_{OPT}) \\
 &\overset{(g)}{\geq} \Big(1- \prod_{k=1}^i \exp \Big(-\sum_{j=1}^d\frac{h_j(w_k)}{H_j}\ \Big)\Big)g(\mathcal{Z}_{OPT}) \\
& =  \Big(1- \exp \Big(-\sum_{k=1}^{i}\sum_{j=1}^d\frac{h_j(w_k)}{H_j}\Big)\Big)g(\mathcal{Z}_{OPT}) \\
 &\overset{(h)}{=}\Big(1- \exp \Big(-\sum_{j=1}^d\frac{h_j(\mathcal{Z}_i)}{H_j}\Big)\Big)g(\mathcal{Z}_{OPT})
\end{split}
\label{eq:inequality5}
\end{equation}
where (g) is due to the inequality 
\begin{equation} \small
1-x \leq \exp(-x), \quad \mbox{for} \quad x > 0
\label{eq:inequality6}
\end{equation}
and (h) results from the fact that $\sum_{k=1}^i h_j(w_k) = h_j(\mathcal{Z}_i)$

Let $k_{i^*}$ be the last step at which line 6 of Algorihtm \ref{algo:weighted_cost_appendix} evaluates to True. Then $\mathcal{Z}_{WCB} = \mathcal{Z}_{k_{i^*}} = \{w_1, w_2, \dots, w_{k_{i^*}}\}$ is the solution returned by the WCB algorithm. Now let $w_{k_{i^*}+1}$ be the element evaluated in line 5 of Algorithm \ref{algo:weighted_cost_appendix} at step $k_{i^*}+1$. By assumption, $w_{k_{i^*}+1}$ violates at least one of the $d$ constraints in line 6. Let $j^*$ be the index of the constraint that is violated. Since when adding the element $w_i$ to the set $\mathcal{Z}_i$ in the above analysis, we did not assume that this element satisfies the constraints in line 6 of Algorithm  \ref{algo:weighted_cost_appendix}, inequality \eqref{eq:inequality5} holds for the set $ \mathcal{Z}_{k_{i^*}} \cup w_{k_{i^*}+1}$, {\em i.e.}
\begin{equation} \small
\begin{split}
g(&\mathcal{Z}_{k_{i^*}}  \cup w_{k_{i^*}+1}) \\
&= g(\mathcal{Z}_{WCB} \cup w_{k_{i^*}+1}) \\
 &> \Big(1- \exp \Big(-\sum_{j=1}^d\frac{h_j(\mathcal{Z}_{WCB} \cup w_{k_{i^*}+1})}{H_j}\Big)\Big)g(\mathcal{Z}_{OPT})\\
&\overset{(i)} {>} (1-\mbox{e}^{-1}) g(\mathcal{Z}_{OPT})
\end{split}
\label{eq:inequality7}
\end{equation}
To obtain inequality (i)  in \eqref{eq:inequality7} we used the fact that, by assumption, $h_{j^*}(\mathcal{Z}_{WCB} \cup w_{k_{i^*}+1}) > H_{j^*}$.
Therefore, the following inequality also holds

\begin{equation} \small
\begin{split}
\sum_{j=1}^d & \frac{h_j(\mathcal{Z}_{WCB} \cup w_{k_{i^*}+1})}{H_j} \\
&= \frac{h_{j^*}(\mathcal{Z}_{WCB} \cup w_{k_{i^*}+1})}{H_{j^*}} 
+ \sum_{j\neq j^*}\frac{h_j(\mathcal{Z}_{WCB} \cup w_{k_{i^*}+1})}{H_j}  \\
& > 1
\end{split}
\label{eq:inequality8}
\end{equation}

Finally, to obtain the approximation ratio of Theorem \ref{th:approx_ratio}, let $w^* \triangleq \argmax _{\{w\in\mathcal{W}: h_j(w) < H_j, \forall \;  j\}}g(w)$. Then, by definition, $g(\mathcal{Z}_{UC}) \geq g(w^*) \geq g(w_{k_{i^*}+1})$. Using the definition of $\mathcal{Z}^*$, we have 
\begin{equation} \small
\begin{split}
g(\mathcal{Z}^*) &\geq \frac{g(\mathcal{Z}_{UC})+ g(\mathcal{Z}_{WCB})}{2}  \\
& \geq \frac{1}{2}\Big( g(w_{k_{i^*}+1}) +g(\mathcal{Z}_{WCB}) \Big) \\
&  \overset{(j)}{\geq} \frac{1}{2} g(\mathcal{Z}_{WCB} \cup w_{k_{i^*}+1}) \\
&  > \frac{1}{2}(1-\mbox{e}^{-1}) g(\mathcal{Z}_{OPT})
\end{split}
\label{eq:inequality9}
\end{equation}
where inequality (j) is due to the subadditivity property of submodular functions.
\end{proof}

\begin{corollary}
Let  $\mathcal{Z}^*  \triangleq \argmax_{\mathcal{Z}\in\{  \mathcal{Z}_{UC}, \mathcal{Z}_{WCB}\}} g(\mathcal{Z})$ when applying the UC and WCB algorithms to the optimization problem in \eqref{eq:submodmax_disjoint}. Then,  $g(\mathcal{Z}^*) > \frac{1}{2}(1-e^{-1}) g(\mathcal{Z}_{OPT})$.
\label{cor:disjoint_approx_ratio}
\end{corollary}
\begin{proof} The proof of Corollary \ref{cor:disjoint_approx_ratio} is similar to the proof of Theorem \ref{th:approx_ratio}. We therefore omit most of the details and highlight only the differences in the two proofs.

Replacing $h_j(w_i)$  and $H_j$ with $\sum_{m=1}^{M_j} h_j^m(w_i)$ and $\sum_{m=1}^{M_j} H_j^m$, respectively, and $d$ with $d^\prime$ in inequality \eqref{eq:inequality1} and following the same procedure that led to \eqref{eq:inequality3}, we obtain
\begin{equation} \small
\begin{split}
g(\mathcal{Z}_{i}) & - g(\mathcal{Z}_{OPT})  \\
&> \Big(1 -\sum_{j=1}^{d^\prime}\frac{\sum_{m=1}^{M_j} h_j^m(w_i)}{\sum_{m=1}^{M_j} H_j^m} \Big)  \Big(  g(\mathcal{Z}_{i-1}) - g(\mathcal{Z}_{OPT})   \Big) \\
& >  \Big(1 -\sum_{j=1}^{d^\prime}\sum_{m=1}^{M_j}\frac{ h_j^m(w_i)}{ H_j^m} \Big)  \Big(  g(\mathcal{Z}_{i-1}) - g(\mathcal{Z}_{OPT}) \Big)
\end{split}
\label{eq:inequality10}
\end{equation}
where the second inequality in \eqref{eq:inequality10} is due to the following inequality
\begin{equation}
\frac{\sum_{i=1}^ma_i}{\sum_{i=1}^mb_i} < \sum_{i=1}^m \frac{a_i}{b_i}, \quad \mbox{for} \quad a_i, b_i > 0
\label{eq:inequality11}
\end{equation}
The result follows immediately by applying the same arguments as those that we used to prove the inequalities \eqref{eq:inequality5}, \eqref{eq:inequality7} and \eqref{eq:inequality9}.
\end{proof}
\bibliographystyle{IEEEtran}

\begin{thebibliography}{10}
\providecommand{\url}[1]{#1}
\csname url@samestyle\endcsname
\providecommand{\newblock}{\relax}
\providecommand{\bibinfo}[2]{#2}
\providecommand{\BIBentrySTDinterwordspacing}{\spaceskip=0pt\relax}
\providecommand{\BIBentryALTinterwordstretchfactor}{4}
\providecommand{\BIBentryALTinterwordspacing}{\spaceskip=\fontdimen2\font plus
\BIBentryALTinterwordstretchfactor\fontdimen3\font minus
  \fontdimen4\font\relax}
\providecommand{\BIBforeignlanguage}[2]{{%
\expandafter\ifx\csname l@#1\endcsname\relax
\typeout{** WARNING: IEEEtran.bst: No hyphenation pattern has been}%
\typeout{** loaded for the language `#1'. Using the pattern for}%
\typeout{** the default language instead.}%
\else
\language=\csname l@#1\endcsname
\fi
#2}}
\providecommand{\BIBdecl}{\relax}
\BIBdecl

\bibitem{CiscoVNI}
``{Cisco Visual Networking Index: Forecast and Methodology, 2016-2021},'' White
  Paper, Cisco Systems Inc., Jun. 2016.

\bibitem{GolrezaeiINFOCOM2012}
N.~Golrezaei, K.~Shanmugam, A.~G. Dimakis, A.~F. Molisch, and G.~Caire,
  ``{FemtoCaching: wireless video content delivery through distributed caching
  helpers},'' in \emph{IEEE INFOCOM'12}, Mar. 2012.

\bibitem{PoularakisTMC2017}
K.~Poularakis and L.~Tassiulas, ``{Code, Cache and Deliver on the Move: A Novel
  Caching Paradigm in Hyper-Dense Small-Cell Networks},'' \emph{IEEE Trans. on
  Mobile Comp.}, vol.~16, no.~3, pp. 675--687, Mar. 2017.

\bibitem{PoularakisTCOMM2014}
K.~Poularakis, G.~Iosifidis, and L.~Tassiulas, ``Approximation algorithms for
  mobile data caching in small cell networks,'' \emph{IEEE Trans. on
  Communications}, vol.~62, no.~10, pp. 3665--3677, Oct. 2014.

\bibitem{DeAbreuJVCIR2015}
A.~D. Abreu, L.~Toni, N.~Thomos, T.~Maugey, F.~Pereira, and P.~Frossard,
  ``Optimal layered representation for adaptive interactive multiview video
  streaming,'' \emph{Journal of Visual Communication and Image Representation},
  vol.~33, pp. 255 -- 264, 2015.

\bibitem{WangIEEEACCESS2017}
W.~Wang, R.~Lan, J.~Gu, A.~Huang, H.~Shan, and Z.~Zhang, ``Edge caching at base
  stations with device-to-device offloading,'' \emph{IEEE Access}, vol.~5, pp.
  6399--6410, 2017.

\bibitem{KhreishahJSAC2016}
A.~Khreishah, J.~Chakareski, and A.~Gharaibeh, ``Joint caching, routing, and
  channel assignment for collaborative small-cell cellular networks,''
  \emph{IEEE JSAC}, vol.~34, no.~8, pp. 2275--2284, Aug. 2016.

\bibitem{OzfaturaCOMMLET2017}
E.~Ozfatura and D.~G\"{u}nd\"{u}z, ``Mobility and popularity-aware coded
  small-cell caching,'' \emph{IEEE Communications Letters}, 2017.

\bibitem{ToniMMSP2013}
L.~Toni, N.~Thomos, and P.~Frossard, ``{Interactive Free Viewpoint Video
  Streaming Using Prioritized Network Coding},'' in \emph{Proc. of IEEE
  MMSP'13}, Sept. 2013, pp. 446--451.

\bibitem{SchmeingBOOK2011}
M.~Schmeing and X.~Jiang, \emph{Depth Image Based Rendering}.\hskip 1em plus
  0.5em minus 0.4em\relax Springer Berlin Heidelberg, 2011, pp. 279--310.

\bibitem{RenARXIV2012}
D.~Ren, S.-H.~G. Chan, G.~Cheung, V.~Zhao, and P.~Frossard, ``Collaborative
  {P2P} streaming of interactive live free viewpoint video,''
  \emph{arXiv:1211.4767v1 [cs.MM]}, 2012.

\bibitem{TzelepisIVC2016}
C.~Tzelepis, Z.~Ma, V.~Mezaris, B.~Ionescu, I.~Kompatsiaris, G.~Boato, N.~Sebe,
  and S.~Yan, ``Event-based media processing and analysis: A survey of the
  literature,'' \emph{Image and Vision Computing}, vol.~53, pp. 3 -- 19, 2016.

\bibitem{BlascoICC14}
P.~Blasco and D.~Gunduz, ``Learning-based optimization of cache content in a
  small cell base station,'' in \emph{Proc. of IEEE Int'l Conf. on
  Communications (ICC)}, 2014.

\bibitem{EkmekciogluTCSVT2017}
E.~Ekmekcioglu, C.~Gurler, A.~Kondoz, and A.~Tekalp, ``Adaptive multiview video
  delivery using hybrid networking,'' \emph{IEEE Transactions on Circuits and
  Systems for Video Technology}, vol.~27, no.~6, pp. 1313--1325, Jun. 2017.

\bibitem{GareyBOOK1979}
M.~R. Garey and D.~S. Johnson, \emph{{Computers and intractability: a guide to
  the theory of {NP}-completeness}}.\hskip 1em plus 0.5em minus 0.4em\relax New
  York, NY, USA: W. H. Freeman \& Co., 1979.

\bibitem{SlijepcevicICC2001}
S.~Slijepcevic and M.~Potkonjak, ``Power efficient organization of wireless
  sensor networks,'' in \emph{Proc. of IEEE Int'l Conf. on Communications},
  Helsinki, Finland, Jun. 2001.

\bibitem{KrauseBOOK2013}
A.~Krause and D.~Golovin, ``Submodular function maximization,'' in
  \emph{Tractability: Practical Approaches to Hard Problems}.\hskip 1em plus
  0.5em minus 0.4em\relax Cambridge University Press, 2013.

\bibitem{LiTCSVT2017}
C.~Li, L.~Toni, J.~Zou, H.~Xiong, and P.~Frossard,
  ``Delay-power-rate-distortion optimization of video representations for
  dynamic adaptive streaming,'' \emph{IEEE Transactions on Circuits and Systems
  for Video Technology}, vol.~PP, no.~99, pp. 1--1, 2017.

\bibitem{MaggiARXIV2012}
L.~Maggi, L.~Gkatzikis, G.~Paschos, and J.~Leguay, ``Adapting caching to
  audience retention rate: Which video chunk to store?''
  \emph{arXiv:1512.03274v1 [cs.NI]}, 2012.

\bibitem{ParisisCOMCOM2017}
G.~Parisis, V.~Sourlas, K.~V. Katsaros, W.~K. Chai, G.~Pavlou, and I.~Wakeman,
  ``Efficient content delivery through fountain coding in opportunistic
  information-centric networks,'' \emph{Comput. Commun.}, vol. 100, no.~C, pp.
  118--128, Mar. 2017.

\bibitem{NemhauserMATHPROG1978}
G.~L. Nemhauser, L.~A. Wolsey, and M.~L. Fisher, ``An analysis of
  approximations for maximizing submodular set functions---i,''
  \emph{Mathematical Programming}, vol.~14, no.~1, pp. 265--294, Dec. 1978.

\bibitem{LeskovecSIGKDD2007}
J.~Leskovec, A.~Krause, C.~Guestrin, C.~Faloutsos, J.~VanBriesen, and
  N.~Glance, ``Cost-effective outbreak detection in networks,'' in \emph{Proc.
  ACM SIGKDD Int'l Conf. on Knowledge Discovery and Data Mining}, New York, NY,
  2007, pp. 420--429.

\bibitem{SviridenkoOPRESLET2004}
M.~Sviridenko, ``A note on maximizing a submodular set function subject to a
  knapsack constraint,'' \emph{{Operations Research Letters}}, vol.~32, no.~1,
  pp. 41--43, Jan. 2004.

\bibitem{KrauseTECHREP2005}
A.~Krause and C.~Guestrin, ``A note on the budgeted maximization of submodular
  functions,'' Carnegie Mellon University, Tech. Rep. CMU-CALD-05-103, 2005.

\end{thebibliography}


\end{document}